\documentclass[a4paper,UKenglish,cleveref, autoref, thm-restate, numberwithinsect]{lipics-v2021}

\title{Pattern Matching with Mismatches and Wildcards}

\date{\empty}

\author{Gabriel Bathie}{DIENS, \'{E}cole normale sup\'{e}rieure de Paris, PSL Research University, France \and LaBRI, Université de Bordeaux, France}{gabriel.bathie@gmail.com}{https://orcid.org/0000-0003-2400-4914}{}

\author{Panagiotis Charalampopoulos}{Birkbeck, University of London, UK}{pcharalampo@gmail.com}{https://orcid.org/0000-0002-6024-1557}{}

\author{Tatiana Starikovskaya}{DIENS, \'{E}cole normale sup\'{e}rieure de Paris, PSL Research University, France}{tat.starikovskaya@gmail.com}{https://orcid.org/0000-0002-7193-9432}{}

\authorrunning{G.~Bathie, P.~Charalampopoulos, and T.~Starikovskaya}

\Copyright{Gabriel Bathie, Panagiotis Charalampopoulos and Tatiana Starikovskaya}



\nolinenumbers 

\hideLIPIcs

\usepackage{xspace}
\usepackage[T1]{fontenc}
\usepackage[utf8]{inputenc}
\usepackage{xcolor}
\usepackage{amsmath}
\usepackage{amssymb}
\usepackage{amsthm}
\usepackage{thmtools}
\usepackage{thm-restate}
\usepackage{ifthen}
\usepackage{tikz}
\usetikzlibrary{calc} 
\usetikzlibrary{decorations.pathreplacing} 
\usetikzlibrary{patterns}
\usetikzlibrary{math} 

\usepackage{xcolor,stackengine,graphicx,scalerel,xspace}
\usepackage{multirow}

\makeatletter
\newcommand*{\rom}[1]{\expandafter\@slowromancap\romannumeral #1@}
\makeatother
\newcommand{\wild}{\protect\scalebox{.7}{\protect\stackinset{c}{}{c}{}{$\lozenge$}{\textcolor{black!10}{$\blacklozenge$}}}}
\def\dd{\mathinner{.\,.}}
\newcommand{\eps}{\varepsilon}
\newcommand{\rot}{\textsf{rot}}
\newcommand{\lcp}{\textsf{lcp}\xspace}
\newcommand{\dontcare}{D}
\newcommand{\gaps}{G}
\newcommand{\Oh}{\mathcal{O}}

\newcommand{\cO}{\mathcal{O}}
\newcommand{\cOtilde}{\tilde{\cO}}
\newcommand{\per}{\textsf{per}}

\newcommand{\norm}[1]{\Vert #1{\Vert}}
\newcommand{\ceil}[1]{\lceil #1 \rceil}
\def\pillar{{\tt PILLAR}\xspace}

\DeclareMathOperator{\polylog}{polylog}

\newcommand{\occhk}{\ensuremath\mathrm{Occ}_k}
\newcommand{\dkgk}{(\dontcare + k)(\gaps + k)}

\newcommand{\M}{\mathcal{M}}
\newcommand{\R}{\mathcal{R}}
\renewcommand{\S}{\mathcal{S}}
\newcommand{\lmisp}{\mathsf{LeftMisper}}
\newcommand{\rmisp}{\mathsf{RightMisper}}
  
\newcommand{\Misp}{\mathsf{MI}}
\newcommand{\Ham}{\delta_H}
\newcommand{\Matching}{\textsf{Matching}}
\newcommand{\Aligned}{\textsf{Aligned}}
\newcommand{\Hidden}{\textsf{Hidden}}

\newtheorem{fact}[theorem]{Fact}

\usepackage{todonotes}

\begin{document}

\maketitle
\begin{abstract}
In this work, we address the problem of approximate pattern matching with wildcards. Given a pattern $P$ of length $m$ containing $D$ wildcards, a text $T$ of length $n$, and an integer $k$, our objective is to identify all fragments of $T$ within Hamming distance $k$ from $P$.

Our primary contribution is an algorithm with runtime $\mathcal{O}(n + (D+k)(G+k)\cdot n/m)$ for this problem. Here, $G \le D$ represents the number of maximal wildcard fragments in $P$. We derive this algorithm by elaborating in a non-trivial way on the ideas presented by~[Charalampopoulos et al., FOCS'20] for pattern matching with mismatches (without wildcards). Our algorithm improves over the state of the art when $D$, $G$, and~$k$ are small relative to $n$. For instance, if $m = n/2$, $k = G = n^{2/5}$, and $D = n^{3/5}$, our algorithm operates in $\mathcal{O}(n)$ time, surpassing the $\Omega(n^{6/5})$ time requirement of all previously known algorithms.

In the case of \emph{exact} pattern matching with wildcards ($k=0$), we present a much simpler algorithm with runtime $\mathcal{O}(n + DG \cdot n/m)$ that clearly illustrates our main technical innovation:
the utilisation of positions of $P$ that do not belong to any fragment of $P$ with a density of wildcards much larger than $D/m$ as anchors for the sought (approximate) occurrences.
Notably, our algorithm outperforms the best-known $\mathcal{O}(n \log m)$-time FFT-based algorithms of [Cole and Hariharan, STOC'02] and [Clifford and Clifford, IPL'04] if $DG = o(m \log m)$.

We complement our algorithmic results with a structural characterization of the $k$-mismatch occurrences of $P$. We demonstrate that in a text of length $\mathcal{O}(m)$, these occurrences can be partitioned into $\mathcal{O}((D+k)(G+k))$ arithmetic progressions. Additionally, we construct an infinite family of examples with $\Omega((D+k)k)$ arithmetic progressions of occurrences, leveraging a combinatorial result on progression-free sets [Elkin, SODA'10].
\end{abstract}

\clearpage
\pagenumbering{arabic}

\section{Introduction}
\label{sec:intro}
Pattern matching is one of the most fundamental algorithmic problems on strings.
Given a text $T$ of length $n$ and a pattern $P$ of length $m$, both over an alphabet $\Sigma$, the goal is to compute all occurrences of $P$ in $T$.
This problem admits an efficient linear-time solution, for example using the seminal algorithm of Knuth, Morris, and Pratt~\cite{knuth1977fast}.
However, looking for exact matches of $P$ in $T$ can be too restrictive for some applications,
for example when working with potentially corrupted data, 
when accounting for mutations in genomic data,
or when searching for an incomplete pattern.
There are several ways to model and to work with corrupt or partial textual data for the purposes of pattern matching.

A natural and well-studied problem is that of computing fragments of the text that are close to the pattern with respect to some distance metric.
One of the most commonly used such metrics is the Hamming distance.
Abrahamson and Kosaraju~\cite{Abrahamson,Kosaraju} independently developed an $\cO(n \sqrt{m \log m})$-time algorithm that computes
the Hamming distance between the pattern and every $m$-length substring of the text using convolutions via the Fast Fourier Transform (FFT).
This complexity has only been recently improved with the state of the art being the randomised $\cO(n \sqrt{m})$-time algorithm of Chan et al.~\cite{DBLP:conf/focs/Chan0WX23} and the deterministic $\cO(n \sqrt{m \log\log m})$-time algorithm of Jin et al.~\cite{DBLP:journals/corr/abs-2403-20326}.
In many applications, one is interested in computing substrings of the text that are
\emph{close} to the pattern instead of computing the distance to every substring.
In this case, an integer threshold $k$ is given as part of the input and the goal is to compute
fragments of $T$ that have at most $k$ mismatches with $P$.
Such a fragment is called a \emph{$k$-mismatch occurrence} of $P$ in $T$.
The state-of-the-art for this problem are the $\cO(n\sqrt{k\log k})$-time algorithm of Amir et al.~\cite{DBLP:journals/jal/AmirLP04},
the $\cO(n + (n/m)\cdot k^2)$-time algorithm of Chan et al.~\cite{DBLP:conf/stoc/ChanGKKP20}, and the $\cO(n + kn/\sqrt{m})$-time algorithm of Chan et al.~\cite{DBLP:conf/focs/Chan0WX23}
that provides a smooth trade-off between the two aforementioned solutions, improving the bound for some range of parameters.
Deterministic counterparts of the last two algorithms (which are randomised) at the expense of extra polylogarithmic factors were presented in~\cite{DBLP:conf/soda/CliffordFPSS16,DBLP:conf/icalp/GawrychowskiU18,unified}.

The structure of the set of $k$-mismatch occurrences of $P$ in $T$ admits an insightful characterisation, shown by Charalampopoulos et al.~\cite{unified} who tightened the result of Bringmann et al.~\cite{DBLP:conf/soda/BringmannWK19}:
either $P$ has $\cO(k\cdot n / m)$ $k$-mismatch occurrences in $T$ or $P$ is at Hamming distance less than $2k$ from a string with period $q=\cO(m/k)$;
further, in the periodic case, the starting positions of the $k$-mismatch occurrences of $P$ in $T$ can be partitioned into $\cO(k^2 \cdot n / m)$ arithmetic progressions with difference $q$.
This characterisation can be exploited towards obtaining efficient algorithms in settings other than the standard one, e.g., in the setting where both $P$ and $T$ are given in compressed form~\cite{unified}, and, in combination with other ideas and techniques, in the streaming setting~\cite{DBLP:conf/focs/KociumakaPS21} and in the quantum setting~\cite{DBLP:conf/soda/JinN23}.

In the case when the positions of the corrupt characters in the two strings are known in advance, one can use a more adaptive approach, by placing a \emph{wildcard} $\wild \not\in \Sigma$, a special character that matches any character in $\Sigma \cup \{ \wild \}$, in each of these positions, and then performing exact pattern matching. 
Already in 1974, Fischer and Paterson~\cite{FP74} presented an $\cO(n\log m \log \sigma)$-time algorithm for the pattern matching problem with wildcards.
Subsequent works by Indyk~\cite{DBLP:conf/focs/Indyk98a}, Kalai~\cite{DBLP:conf/soda/Kalai02a}, and Cole and Hariharan~\cite{DBLP:conf/stoc/ColeH02} culminated in an $\cO(n\log m)$-time deterministic algorithm~\cite{DBLP:conf/stoc/ColeH02}.
A few years later, Clifford and Clifford~\cite{DBLP:journals/ipl/CliffordC07} presented a very elegant algorithm with the same complexities.
All the above solutions are based on fast convolutions via the FFT.

Unsurprisingly, the pattern matching problem in the case where we both have
wildcards and allow for mismatches has also received significant attention.
Conceptually, it covers the case where some of the corrupt positions are known, but not all of them.
We denote by $D$ the total number of wildcards in $P$ and $T$, and by $G$ the number of maximal fragments in $P$ and $T$
all of whose characters are wildcards.
A summary of known results for the considered problem is provided in \cref{tab:results}.

Note that, in practice, $G$ may be much smaller than $D$. 
For example, DNA sequences have biologically important loci, which are characterised using the notion of structured motifs~\cite{biology}:
sequences of alternating conserved and non-conserved blocks.
Conserved blocks are ones which are identical across intra- or inter-species occurrences of the structured motif, while non-conserved ones are not known to have biological significance and can vary significantly across such occurrences.
Non-conserved blocks can be hence modelled with blocks of wildcards as in \cite{DBLP:journals/ipl/ManberB91}.
In this case, evidently, we have $G$ being much smaller than~$D$.
This feature has been used in the literature before, e.g., for the problem of answering longest common compatible prefix queries over a string with wildcards.
Crochemore et al.~\cite{DBLP:journals/jda/CrochemoreIKKLR15} showed an $\cO(n\gaps)$-time construction algorithm for a data structure that is capable of answering such queries in $\cO(1)$ time, while the previously best known construction time was $\cO(n\dontcare)$~\cite{DBLP:conf/lata/Blanchet-SadriL13}.

In several applications, it is sufficient to only account for wildcards in one of $P$ and $T$:
in the application we just discussed, the text is a fixed DNA sequence, whereas the sought pattern, the structured motif, is modelled as a string with wildcards.
In such cases, one can obtain more efficient solutions than those for the general case where both $P$ and $T$ have wildcards,
such as the ones presented in \cite{DBLP:journals/ipl/CliffordP10,DBLP:journals/ipl/NicolaeR17} and the one we present here.

\renewcommand{\arraystretch}{1.4}
\begin{table}[t!]
    \caption{Previous results on pattern matching with wildcards under the Hamming distance.}\label{tab:results}
    \centering
    \begin{tabular}{c|c|c}
         Time complexity & $\wild$ in & Reference \\\hline
         $\cO(nk^2 \log^2 m)$ & \multirow{5}{*}{$P$ and $T$} & \cite{DBLP:journals/jcss/CliffordEPR10} \\\cline{1-1}\cline{3-3}
         $\cO(n(k + \log m \log k) \log n)$ &  & \cite{DBLP:journals/jcss/CliffordEPR10} \\\cline{1-1}\cline{3-3}
         $\cO(nk \polylog m)$ &  & \cite{DBLP:conf/soda/CliffordEPR09,DBLP:journals/jcss/CliffordEPR10,DBLP:journals/algorithms/NicolaeR15} \\\cline{1-1}\cline{3-3} 
         $\cO(n\sqrt{m\log m})$  & & follows from~\cite{Abrahamson,Kosaraju}, cf.~\cite{DBLP:journals/jcss/CliffordEPR10} \\\cline{1-1}\cline{3-3}
         $\cO(n\sqrt{m-D}\log m)$  & & \cite{DBLP:journals/jal/AmirLP04} \\\hline
         $\cO(n\sqrt[3]{mk\log^2 m})$ & one of $P$ or $T$ & \cite{DBLP:journals/ipl/CliffordP10} \\\hline
         $\cO(n\sqrt{k\log m} + n \cdot \min\{\sqrt[3]{Gk\log^2 m}, \sqrt{G\log m}\})$ & $P$ & \cite{DBLP:journals/ipl/NicolaeR17} \\
    \end{tabular}
\end{table}

\subparagraph{Multi-framework algorithms with the \pillar model.}
We describe our algorithms in the \pillar model,
introduced by Charalampopoulos et al.~\cite{unified}.
In this model, we account for the number of calls to a small set of versatile primitive operations on strings, called \pillar operations,
such as longest common extension queries or internal pattern matching queries, plus any extra time required to perform usual word RAM operations. 
The \pillar model allows for a unified approach across several settings, due to known efficient implementations of \pillar operations.
These settings include the standard word RAM model, 
the compressed setting, where the strings
are compressed as straight-line programs,
the dynamic setting~\cite{gawrychowski2018optimal}, and the quantum setting.
Therefore, in essence, we provide meta-algorithms, that can be combined with efficient \pillar implementations to give efficient algorithms for a variety of settings.

\subparagraph{The standard trick.}
For reasons related to the periodic structure of strings, it is often convenient to assume
that the length of the text is at most $3m/2$, where $m$ is the length of the pattern.
This does not pose any actual restrictions as one can cover $T$ with $\cO(n/m)$
fragments, each of length $3m/2$ (except maybe the last one) such that each two consecutive fragments overlap on $m-1$ positions.
Then, any occurrence of $P$ in $T$ is contained in exactly one of these fragments.
Thus, an algorithm with runtime $C(m)$
for a pattern of length $m$ and a text of length at most $3m/2$,
readily implies an algorithm with runtime $\cO(C(m) \cdot n/m)$ for texts of length $n$,
as one can run $\cO(n/m)$ separate instances and aggregate the results.

\subparagraph{Reduction to pattern matching with mismatches.}
The problem of $k$-mismatch pattern matching with $\dontcare$ wildcards
can be straightforwardly reduced to $(\dontcare+k)$-mismatch pattern matching
in \emph{solid strings}, i.e., strings without wildcards.
In what follows we consider solid texts.
Given the pattern $P$, construct the string $P_\#$ obtained by replacing every
wildcard in~$P$ with a new character $\# \notin \Sigma$.
Observe that
  a pattern $P$ with $\dontcare$ wildcards
  has a $k$-mismatch occurrence at a position~$i$ of a solid text $T$ if and only if
  $P_\#$ has a $(\dontcare + k)$-mismatch occurrence at that position.

In~\cite{unified}, the authors present an efficient algorithm for
the $d$-mismatch pattern matching problem for solid strings in the \pillar model.

\begin{theorem}[{\cite[Main Theorem 8]{unified}}]\label{fact:thm-unified}
  Let $S$ and $T$ be solid strings of respective lengths $m$ and $n \le 3m/2$.
  We can compute a representation of the $d$-mismatch occurrences of $S$
  in $T$ using $\cO(d^2\log\log d)$ time plus $\cO(d^2)$ \pillar operations.
\end{theorem}
Applying \cref{fact:thm-unified} with $S= P_\#$ and $d = \dontcare + k$,
we obtain an algorithm for $k$-mismatch pattern matching
with $\dontcare$ wildcards that runs in $\cOtilde((\dontcare + k)^2)$ time in the \pillar model.\footnote{In this work, the notation $\tilde{\mathcal{O}}(\cdot)$ suppresses factors polylogarithmic in the length of the input strings.}

\subparagraph{Our results.}
We provide a more fine-grained result, replacing one $\dontcare$ factor with a~$\gaps$ factor.
We also make an analogous improvement over the structural result for the set of $k$-mismatch occurrences obtained via the reduction to $(D+k)$-mismatch pattern matching.
Our main result can be formally stated as follows.

\begin{restatable}{theorem}{hampm}\label{thm:ham-pm}
  Let $P$ be a pattern of length $m$ with $\dontcare$ wildcards arranged
  in $\gaps$ groups, $T$ be a solid text of length $n\leq 3m/2$, and $k$ be a positive integer.
  We can compute a representation of the $k$-mismatch occurrences of $P$
  in $T$ as $\cO((\dontcare + k)\gaps)$ arithmetic progressions
  with common difference
  and $\cO((\dontcare + k)k)$ additional occurrences
  using $\cO((\dontcare + k)\cdot (\gaps + k)\log\log (\dontcare + k))$
  time plus $\cO((\dontcare + k)\cdot (\gaps + k))$ \pillar operations.
\end{restatable}

In the usual word RAM model, by using known implementations of the \pillar operations 
with $O(n)$ preprocessing time and $O(1)$ operation time,
using the ``standard trick'', and observing that the loglogarithmic factor can be avoided at the cost of $\cO(n)$ extra time,
we obtain an algorithm with runtime $\cO(n + (n/m)(\dontcare + k)(\gaps + k))$
for texts of arbitrary length~$n$.
\cref{sec:pillar} details the implementation of our algorithm in other settings, such as the dynamic and compressed settings.
For example, given a solid text $T$ and a pattern $P$ with~$\dontcare$ wildcards represented as straight-line programs of sizes $N$ and~$M$ respectively,
we can compute the number of $k$-mismatch occurrences of $P$ in $T$ in time $\cOtilde(M + N \cdot (\dontcare + k)(\gaps + k))$, without having to uncompress $P$ and $T$.

We complement our structural result with a lower bound on the
number of arithmetic progressions of occurrences of a pattern with mismatches and wildcards (\cref{thm:lower-bound}), based on a neat construction that employs large sets that do not contain any arithmetic progression of size 3~\cite{https://doi.org/10.1112/jlms/s1-11.4.261,Behrend,DBLP:conf/soda/Elkin10}. Informally, we show that there exist a pattern $P$ and a text $T$ of length at most $3|P|/2$ such that the set of $k$-mismatch occurrences of $P$ in~$T$
  cannot be covered with less than $\Omega((\dontcare+k) \cdot (k+1))$
  arithmetic progressions. 
This implies, in particular, a lower bound of $\Omega(\dontcare)$ on the number of arithmetic progressions of exact occurrences for a pattern with $\dontcare$ wildcards and a lower bound of $\Omega(k^2)$ on the number of arithmetic progressions of $k$-mismatch occurrences of a solid pattern, thus showing the tightness of the known upper bound~\cite{unified}.

When $k = 0$, \cref{thm:ham-pm} readily implies an $\cO(n+\dontcare \gaps \cdot n/m)$-time algorithm for \emph{exact} pattern matching. However, the techniques employed by this algorithm are rather heavy-handed, and this can be avoided. In \cref{sec:exact-pm}, we present a much simpler algorithm that achieves the same time complexity and showcases the primary technical innovation of our approach: the utilization of carefully selected positions, termed \emph{sparsifiers}, which exclusively belong to fragments $F$ of $P$ such that the ratio of the number of wildcards within them to their length is bounded by $\cO(D/m)$. In the standard word RAM model, the implied $\cO(n+\dontcare \gaps \cdot n/m)$ time complexity for exact pattern matching outperforms the state-of-the-art $\cO(n\log m)$~\cite{DBLP:journals/ipl/CliffordC07,DBLP:conf/stoc/ColeH02} when $DG = o(m\log m)$.

\subparagraph{Technical overview.}
To illustrate how sparsifiers help, consider our algorithm for exact pattern matching, which
draws ideas from the work of Bringmann et al.~\cite{DBLP:conf/soda/BringmannWK19}.
We first compute a solid $\Omega(m/G)$-length fragment $S$ of $P$ that contains a sparsifier.
We then compute its exact matches in $T$.
If $S$ only has a few occurrences, we straightforwardly verify which of those extend to occurrences of $P$.
However, if $S$ has many occurrences we cannot afford to do that, and we instead have to exploit the implied periodic structure of $S$.
We distinguish between two cases.
In the case when $P$ matches a periodic string with the same period as~$S$, denoted $\per(S)$,
we take a sliding window approach as in~\cite{DBLP:conf/soda/BringmannWK19}, using the fact that the wildcards are organised
in only $\gaps$ groups.
The remaining case poses the main technical challenge. 
In that case, our goal is to align the maximal fragment $S':=P[i \dd j]$ of $P$ that contains $S$ and matches a solid string with period $\per(S)$
with a periodic fragment of $T$ such that position $i-1$ is aligned with a position breaking the periodicity in $T$; a so-called \emph{misperiod}.
To this end, we compute $\cO(G)$ maximal fragments of $T$, called $S$-runs, that have period $\per(S)$.
The issue is, however, that up to $D$ misperiods in $T$ might be aligned with wildcards of $S'$.
A straightforward approach would be to extend each $S$-run to the left, allowing for $D+1$ misperiods
and to try aligning each such misperiod with $i-1$. This would yield an algorithm with
runtime $\cO(\gaps^2 \dontcare)$ in the \pillar model, as we would have $\cO(DG)$ candidate misperiods to align position $i-1$ with, and the verification time for each such alignment is $\cO(G)$.
The crucial observation is that since $S'$ contains a sparsifier, we do not need to extend each $S$-run allowing for $D+1$ misperiods.
Instead, we extend it while the ratio of the encountered misperiods to its length does not exceed $20 \cdot \dontcare/m$.
By skipping $S$-runs that are covered due to the extension of other $S$-runs, we ensure that the total number of misperiods with which we
align $i-1$ is only $\cO(D)$, obtaining the desired complexity.

As for handling mismatches, we follow the framework of Charalampopoulos et al.~\cite{unified} for $k$-mismatch pattern matching on solid strings.
They showed that an efficient structural analysis of a solid pattern can return a number of so-called \emph{breaks} or a number of so-called \emph{repetitive regions}, or conclude that the pattern is \emph{almost periodic}.
They treated each of the three cases separately, exploiting the computed structure.
We make several alterations to account for wildcards, such as ensuring that breaks are solid strings and adapting the sliding window approach.
The primary technical challenge in achieving an efficient solution lies in limiting the number of occurrences of repetitive regions in $T$. The greater the number of repetitive region occurrences, the higher the number of potential starting positions for $k$-mismatch occurrences of $P$. 
We achieve that by ensuring that each repetitive region contains a sparsifier.
This way, we force an upper bound on the number of wildcards in each repetitive region, which, in turn,
allows us to bound the number of its approximate occurrences in $T$.

Our lower bound is based on a neat construction that employs large sets that do not contain any arithmetic progression of size 3~\cite{https://doi.org/10.1112/jlms/s1-11.4.261,Behrend,DBLP:conf/soda/Elkin10}.
We use these sets to construct the pattern $P$ and the text $T$. They consist mostly of $0$s, except that $P$ contains wildcards and $1$s positioned at indices that form a progression-free set, and $T$ contains $0$s also positioned at indices that form a progression-free set, but that are far apart from each other. Our construction ensures that there is a $k$-mismatch occurrence of $P$ at position $i$ in $T$
if and only if a $1$ or a wildcard of $P$ is aligned with a $1$ of $T$.
We show that the set of such positions~$i$ has size $\Omega((\dontcare + k) \cdot (k+1))$ and is progression-free.

\subparagraph{Organisation of the paper.}
\cref{sec:prel} introduces concepts relevant to this work, as well as an abstract problem that we use in our analysis.
\cref{sec:exact-pm} presents an algorithm for exact pattern matching with wildcards in the \pillar model. 
It illustrates some of the main ideas underlying this work,
without having to handle mismatches, which bring further challenges.
In \cref{sec:ham-pm}, we describe the algorithm
that underlies \cref{thm:ham-pm}; the implications of this theorem in different settings are provided in \cref{sec:pillar}.
We conclude with \cref{sec:lower-bound} where we present our structural lower bound.

\section{Preliminaries}\label{sec:prel}
In this work, $\Sigma$ denotes an alphabet that consists of integers polynomially bounded in the length of the input strings.
The elements of $\Sigma$ are called \emph{(solid) characters}.
Additionally, we consider a special character denoted by $\wild$ that is not in $\Sigma$ and is called a \emph{wildcard}.
Let $\Sigma_{\wild} = \Sigma \cup \{\wild\}$.
Two characters match if (a) they are identical or (b) at least one of them is a wildcard.
Two equal-length strings match if and only if their $i$-th characters match for all $i$.

For an integer $n\geq 0$, we denote the set of all length-$n$ strings over an alphabet $A$ by~$A^n$.
The set of all strings over $A$ is denoted by $A^*$.
The unique empty string is denoted by $\eps$. 
A string in $\Sigma_{\wild}^*$ is called \emph{solid} if it only contains solid characters, i.e., it is in $\Sigma^*$.

For two strings $S,T \in\Sigma_{\wild}^*$, we use $ST$ to denote their concatenation.
For an integer $m > 0$, the string obtained by concatenating $m$ copies of $S$ is denoted by~$S^m$.
We denote by~$S^\infty$ the string obtained by concatenating infinitely many copies of $S$.

For a string $T \in\Sigma_{\wild}^n$ and an index $i\in [1\dd n]$,\footnote{For integers $i,j\in \mathbb{Z}$, denote $[i\dd j] =
\{k \in \mathbb{Z} : i \le k \le j\}$, $[i
\dd j)=\{k \in \mathbb{Z} : i \le k <
j\}$.} the $i$-th character of $T$ is denoted by~$T[i]$.
We use $|T| = n$ to denote the length of $T$.
For indices $1 \le i, j\le n$, 
$T[i\dd j]$ denotes the \emph{fragment} $T[i] T[{i+1}]\cdots T[j]$ of~$T$ if $i\le j$
and the empty string otherwise. We extend this notation in a natural way to $T[i \dd j+1) = T[i\dd j] = T(i-1 \dd j]$.   
When $i=1$ or $j=n$, we omit these indices, i.e., $T[\dd j] = T[1\dd j]$
and $T[i\dd ] = T[i\dd n]$.   
A string $P$ is a \emph{prefix} of $T$ if there exists $j\in [1\dd n]$
such that $P = T[\dd j]$, and a \emph{suffix} of $T$ if there exists $i\in [1\dd n]$
such that $P = T[i\dd ]$. A \emph{ball} with a radius $r$ and a center $i$, $B_T(i,r)$ is a fragment $T[\max\{1,i-r\}\dd \min\{i+r,n\}]$, where we often omit the subscript $T$ if it is clear from the context. 
 A position $i$ of a string $T$ is called an \emph{occurrence} of a string $P$ if $T[i\dd i+|P|) = P$. 

A positive integer $\rho$ is a \emph{period} of a (solid) string $T\in \Sigma^n$ if $T[i] = T[i+\rho]$ for all $i \in [1 \dd n-\rho]$.
The smallest period of $T$ is referred to as \emph{the period} of $T$ and is denoted by~$\per(T)$. If $\per(T) \leq |T|/2$, 
$T$ is called \emph{periodic}. We exploit the following folklore fact, which is a straightforward corollary of the Fine--Wilf periodicity lemma~\cite{fine1965uniqueness}: 

\begin{corollary}[folklore]\label{cor:occurrences}
Let $P,T \in \Sigma^*$ be solid strings such that $|T| < 2 |P|$.
The set of occurrences of $P$ in $T$ can be represented as one arithmetic progression (possibly, trivial) with difference equal to $\per(P)$. 
\end{corollary} 

For an integer $s \in [1 \dd n)$, we denote by $\rot^s(T)$ the string $T[s+1]\cdots T[n]T[1] \cdots T[s]$, while $\rot^0(T):=T$.
More generally, for any $s \in \mathbb{Z}$, we denote by $\rot^s(T)$ the string $\rot^{x}(T)$, where $x = s \bmod{|T|}$.
A non-empty (solid) string is called \emph{primitive} if it is different from each of its non-trivial rotations.

\begin{fact}[\cite{AlgorithmsOnStrings}]\label{fact:period_primitive}
For any solid string $T\in \Sigma^n$, the prefix $T[1 \dd \per(T)]$ is primitive. 
\end{fact}

An integer $\rho$ is a \emph{deterministic period} of a string $S \in \Sigma_{\wild}^*$ (that may contain wildcards), if there exists a solid string $T$ that matches $S$ and has period $\rho$.

The Hamming distance $\Ham(S_1,S_2)$ between two equal-length strings $S_1,S_2$ in $\Sigma_{\wild}^*$ is the number of positions $i$ such that $S_1[i]$ does not match $S_2[i]$.
For two strings $U, Q \in \Sigma_{\wild}^*$, we slightly abuse notation and denote $\Ham(U, Q^\infty[\dd |U|])$ by $\Ham(U,Q^\infty)$. A position $i$ of a string~$T$ is called a $k$-mismatch occurrence of a string $P$ if $\Ham(T[i,i+|P|),P) \le k$, and the set of all $k$-mismatch occurrences of $P$ in $T$ is denoted by $\occhk(P, T)$.

\subparagraph{The \pillar model.}\label{subsec:pillar}
The \pillar model of computation, introduced in~\cite{unified},
abstracts away the implementation of a versatile set of primitive operations on strings.
In this model, one is given a family of solid strings $\mathcal{X}$ for preprocessing.
The elementary objects are fragments $X[i\dd j]$ of strings $X \in \mathcal{X}$. Each such fragment $S$ is represented via a handle, which is how~$S$ is passed as input to \pillar operations.
Initially, the model provides a handle to each $X \in \mathcal{X}$.
Handles to other fragments can be obtained through an $\mathsf{Extract}$ operation:
\begin{itemize}
\item $\mathsf{Extract}(S,\ell,r)$: Given a fragment $S$ and positions $1 \le \ell \le r \le |S|$, extract
$S[\ell \dd r ]$.
\end{itemize}
Furthermore, given elementary objects $S, S_1, S_2$ the following primitive operations are supported in the \pillar model:
\begin{itemize}
\item $\mathsf{Access}(S,i)$: Assuming $i \in [1\dd |S|]$, retrieve $S[i]$.
\item $\mathsf{Length}(S)$: Retrieve the length $|S|$ of $S$.
\item $\lcp(S_1,S_2)$: Compute the length of the longest common prefix of $S_1$ and $S_2$.
\item $\lcp^R(S_1,S_2)$: Compute the length of the longest common suffix of $S_1$ and $S_2$.
\item Internal pattern matching $\mathsf{IPM}(S_1,S_2)$: Assuming that $|S_2| < 2|S_1|$, compute the set of the starting positions
of occurrences of $S_1$ in $S_2$ represented as one arithmetic progression.
\end{itemize}

We use the following facts; the proof of the second one is provided in \cref{sec:kangaroo}.

\begin{fact}[{\cite[proof of Lemma 12]{DBLP:journals/jcss/Charalampopoulos21}}]\label{fact:infty}
The value $\lcp(X^\infty,Z)$ for a fragment $X$ and a suffix~$Z$ of a solid string $Y$ can be computed in $\cO(1)$ time in the \pillar model.
\end{fact}

\begin{restatable}{fact}{kangaroo}\label{fact:kangaroo}
Let $P$ be a pattern with $D$ wildcards arranged in $G$ groups and $T$ be a solid text.
  For a position $p$ and a given threshold $k \geq 0$, one can test
  whether $\Ham(P, T[i \dd i+m)) \leq k$ in $\cO(\gaps+k)$ time in the \pillar model.
\end{restatable}

We work in the \pillar model despite considering strings with wildcards.
We circumvent this by replacing each wildcard with a solid character $\# \not\in \Sigma$ and using \pillar operations over the obtained collection of (solid) strings.
We ensure that for each string in the collection we can efficiently compute a linked list that stores the endpoints of groups of wildcards.

\subparagraph*{Sparsifiers.}
In \cref{sec:intro}, we elucidated the pivotal role of the fragments of the pattern where wildcards exhibit a ``typical'' distribution. In this section, we formalize this concept.

\begin{definition}[Sparsifiers]
Consider a string $P \in \Sigma_{\wild}^m$ containing $D$ wildcards. We call a position $i$ in $X$ a \emph{sparsifier} if $X[i]$ is a solid character and, for any $r$, the count of wildcards within the ball of radius $r$ centered at $i$ is at most $8r \cdot D/m$. 
\end{definition}

In the following, we demonstrate that $P$ contains a long fragment whose every position is a sparsifier. We start with an abstract lemma, where one can think of a binary vector~$V$ as the indicator vector for wildcards, and $\norm{V}$ denotes the number of 1s in $V$. A \emph{run of~1s (resp.~0s)} is a maximal fragment that consists only of~1s (resp.~0s).

\begin{lemma}\label{lem:toy_problem_better}
Let $V$ be a binary vector of size $N$, $M:=\norm{V}$ and $R$ be the number of runs of 1s in $V$. Assume $V$ to be represented as a linked list of the endpoints of runs of 1s in~$V$, arranged in the sorted order. There is an $\cO(R)$-time algorithm that computes a set $U \subseteq [1\dd N]$ satisfying each of the following conditions:
\begin{enumerate}
\item $|U| \ge N/2-M$,
\item $U$ can be represented as a union of at most $R+1$ disjoint intervals,
\item for each $i \in U$ and radius $r \in [1 \dd N]$, $\norm{B_V(i,r)} \leq 8r \cdot M/N$. 
\end{enumerate}  
\end{lemma}
\begin{proof}
First, we scan $V$ from left to right. Each time we see a 1, we mark the $N/(4M)$ leftmost 0s that are to the right of the considered 1 and have not been already marked. If we mark the last 0 in $V$, we terminate the scan. Overall, we mark at most $N/4$~0s. Next, we perform the symmetric procedure in a right-to-left scan of $V$. 

Let $U$ be the set of positions of unmarked 0s after these two marking steps. We show that $U$ satisfies the condition of the lemma. First, we have $|U| \ge N-M-N/2 \ge N/2-M$. Secondly, by construction, every run of 0s contains at most one interval of unmarked positions and hence $U$ can be represented as union of at most $R+1$ disjoint intervals. Finally, fix $i\in U$ and $r \in [1\dd N]$. Let $B_1 = V[\max\{i-r,1\} \dd i-1]$ and $B_2 = V[i+1 \dd \min\{i+r,N\}]$. Since $V[i]$ was not marked in the left-to-right scan, there are at least $\norm{B_1} \cdot N/(4M)$ 0s in $B_1$.
Symmetrically, since $V[i]$ was not marked in the right-to-left scan of $B$, there are at least $\norm{B_2} \cdot N/(4M)$ 0s in $B_2$. On the other hand, the number of 0s in $B_1$ (resp. $B_2$) is bounded by $r$, and therefore
$\norm{B_V[i,r]} =  \norm{B_1}+\norm{B_2} \le 4r \cdot M/N + 4r \cdot M/N \leq 8r \cdot 4M / N$.

It remains to show that the algorithm can be implemented efficiently. In addition to the linked list $L_1$ representing the runs of 1s in $V$, the algorithm maintains a linked list $L_2$ of the intervals of marked 0s, sorted by their left endpoints. The algorithm simulates marking the~0s for all 1s in the current run at once, taking $\cO(R)$ time in total. Having computed~$L_2$, the algorithm scans $L_1$ and $L_2$ in parallel to extract the set $U$, which takes $\cO(R)$ time as well, thus completing the proof.
\end{proof}

\begin{corollary}\label{cor:sparsifiers}
Consider a string $P \in \Sigma_{\wild}^m$ containing $\dontcare$ wildcards arranged in $\gaps$ groups. If $\dontcare < m/4$, then there is a fragment $S$ of $P$ of length $L = \lfloor m/(8 \gaps) \rfloor$ whose every position is a sparsifier, and one can compute $S$ in $\cO(\gaps)$ time.
\end{corollary}
\begin{proof}
An application of \cref{lem:toy_problem_better} to $P$ with wildcards treated as 1s and solid characters treated as 0s
returns $m/2 - \dontcare > m/4$ sparsifiers in the form of $\gaps+1$ intervals in $\cO(\gaps)$ time.
Thus, there is a fragment of size at least $m/(4(\gaps + 1)) \geq L$ whose every position is a sparsifier;
the $L$-length prefix of this fragment satisfies the condition of the claim. 
\end{proof}

\section{Exact Pattern Matching in the \pillar Model}\label{sec:exact-pm}

In this section, we consider a pattern $P$ of length $m$ with $\dontcare \geq 1$ wildcards arranged in $\gaps$ groups and a solid text $T$ of length $n$ such that $n \leq 3m/2$.
We then use the ``standard trick'' presented in the introduction to lift the result
to texts of arbitrary length.
We prove a structural result for the exact occurrences of $P$ in $T$ and show how to compute them efficiently when the product of $\dontcare$ and $\gaps$ is small.
In particular, we compute them in linear time when $\dontcare \gaps = \cO(m)$, thus improving by a logarithmic factor over the state-of-the-art $\cO(n \log m)$-time algorithms in this case.

\begin{definition}[Misperiods]
Consider a string $V$ over alphabet $\Sigma_{\wild}$.
We say that a position~$x$ is a misperiod with respect to a solid fragment $V[i \dd j]$
when $V[x]$ does not match $V[y]$, where~$y$ is any position in $[i \dd j]$ such that $\per(V[i \dd j])$ divides $|y-x|$.
Additionally, we consider positions~$0$ and $|V|+1$ as misperiods.
We denote the set of the at most~$k$ rightmost misperiods smaller than $i$ with respect to $V[i \dd j]$ by $\lmisp(V,i,j,k)$.
Similarly, we denote the set of the at most $k$ leftmost misperiods larger than $j$ with respect to $V[i \dd j]$ by $\rmisp(V,i,j,k)$.
\end{definition}

\begin{example}
Consider string $V=\texttt{cc\wild bd\underline{\color{red}{abcabcab}}cab}$.
The misperiods with respect to the underlined and highlighted fragment $V[6 \dd 13]$, which has period $3$, are positions $0$, $1$, $5$, and $|V|+1 = 17$.
We have $\lmisp(V,6,13,2) = \{1, 5\}$ and $\rmisp(V,6,13,2) = \{17\}$.
\end{example}

The next lemma states that the sets $\lmisp(V,i,j,k)$ and $\rmisp(V,i,j,k)$ can be computed efficiently in an incremental fashion.
Its proof, which can be found in \cref{sec:mispers_comp}, uses the kangaroo method and closely follows~\cite{DBLP:conf/soda/BringmannWK19,DBLP:journals/jcss/Charalampopoulos21}.

\begin{restatable}{lemma}{computeMispers}\label{lem:compute_mispers}
Consider a string $V$ over an alphabet $\Sigma_{\wild}$ and a solid periodic fragment $V[i \dd j]$ of~$V$.
The elements of either of $\lmisp(V,i,j,|V|)$ and $\rmisp(V,i,j,|V|)$ can be computed in the increasing order with respect to their distance from position $i$ so that:
\begin{itemize}
\item the first misperiod $x$ can be computed in $\cO(1+\gaps_0)$ time in the \pillar model, where $\gaps_0$ denotes the number of groups of wildcards between positions $x$ and $i$;
\item given the $t$-th misperiod $x \not\in \{0, |V|+1\}$, the $(t+1)$-th misperiod can be computed in $\cO(1+\gaps_t)$ time in the \pillar model, where $\gaps_t$ denotes the number of groups of wildcards between said misperiods.
\end{itemize}
\end{restatable}

A direct application of the above lemma yields the following fact.

\begin{corollary}\label{fact:compute_mispers}
For any integer $k$, the sets $\lmisp(V,i,j,k)$ and $\rmisp(V,i,j,k)$ can be computed in $\cO(k+\gaps)$ time in the \pillar model.
\end{corollary}

\begin{definition}
For two strings $S$ and $Q$, let $\Misp(S, Q)$ denote the set of positions of mismatches between $S$ and $Q^\infty$.
\end{definition}

\begin{definition}[$S$-runs]
A fragment of a solid string $V$ spanned by a set of occurrences of a solid string $S$ in $V$
whose starting positions form an inclusion-maximal arithmetic progression with difference $\per(S)$ is called an $S$-run.
\end{definition}

\begin{example}
Let $V = \texttt{cab\underline{\color{red}{abcabcabcab}}c}$ and $S = \texttt{abcab}$.
The underlined and highlighted fragment $V[4 \dd 14]$ is the sole $S$-run in $V$; it is spanned by the occurrences of $S$
at positions $4$, $7$, and $10$.
\end{example}

The following fact characterises the overlaps of $S$-runs; its proof is deferred to \cref{sec:mispers_comp}.

\begin{restatable}{fact}{sruns}\label{fact:sruns}
Two $S$-runs can overlap by no more than $\per(S)-1$ positions.
\end{restatable}

We need a final ingredient before we prove the main theorem of this section.
We state a more general variant of the statement than we need here that also accounts for $k$ mismatches, as this will come handy in the subsequent section.
For the purposes of this section one can think of $k$ as $0$.
The following corollary follows from \cite[Lemma 4.6]{unified} via the reduction to computing $(D+k)$-mismatch occurrences of $P_\#$ in~$T$.

\begin{corollary}[{of~\cite[Lemma 4.6]{unified}}]\label{fact:relevant-fragment}
  Let $S$ be a string of length $m$ with $\dontcare$ wildcards, let $T$ be a solid string such that $|T|\le 3|S|/2$,
  let $k \in [0 \dd m]$ and $d \ge 2(\dontcare + k)$ be a positive integer,
  and let $Q$ be a primitive solid string such that $|Q| \le m/8d$
  and $\Ham(S, Q^\infty) \le d$.
  Then, we can compute, in $\cO(d)$ time in the \pillar model, a fragment $T' = T[\ell \dd r]$ of $T$ such that 
  \begin{itemize}
    \item $\Ham(T', Q^\infty) \le 3d$, and
    \item all elements of $\occhk(S,T') = \{p-\ell : p \in \occhk(S,T)\}$ are equivalent to $0 \pmod{|Q|}$.
  \end{itemize}
\end{corollary}

\begin{theorem}\label{thm:exact}
  Consider a pattern $P$ of length $m$ with $\dontcare$ wildcards arranged in $\gaps$ groups and a solid text $T$ of length $n \leq 3m/2$.
  Either $P$ has $\cO(\dontcare)$ occurrences in~$T$ or~$P$ has a deterministic period $q = \cO(m/\dontcare)$.
  A representation of the occurrences of $P$ in $T$ can be computed in $\cO(\dontcare \gaps \log \log \dontcare)$ time plus the time required to perform $\cO(\dontcare \gaps )$ \pillar operations.
  In the former case the occurrences are returned explicitly, while in the latter case they are returned as $\cO(\dontcare \gaps)$ arithmetic progressions with common difference $q$.
\end{theorem}
\begin{proof}
First, observe that if $\dontcare = \Theta(m)$ the statement holds trivially as there can only be $\cO(m)$ occurrences and we can compute them using $\cO(m\gaps)$ \pillar operations, e.g., by applying \cref{fact:kangaroo} for each position of the text.
We thus henceforth assume that $\dontcare < m/4$.

We apply \cref{cor:sparsifiers} to $P$, thus obtaining, in $\cO(G)$ time, a fragment $S = P[x \dd y]$ of length $m/(8 \gaps)$ whose every position is a sparsifier. (As an implication, $S$ is a solid fragment.) 
Then, we compute all occurrences of $S$ in $T$ in $\cO(\gaps)$ time in the \pillar model, represented as $\cO(\gaps)$ arithmetic progressions with common difference $\per(S)$ (see \cref{cor:occurrences}).

\subparagraph{Case (\rom{1}): $S$ has less than $384 \dontcare$ occurrences in $T$.} In this case, we try to extend each such occurrence to an occurrence of $P$ in $T$ using \cref{fact:kangaroo} in $\cO(\gaps)$ time in the \pillar model. This takes $\cO(\dontcare \gaps)$ time in total in the \pillar model.

\subparagraph{Case (\rom{2}): $S$ has at least $384 \dontcare$ occurrences in $T$.}
In this case, we have two occurrences of $S$ in~$T$ starting within $(3m/2)/(384\dontcare)$ positions of each other, and hence $\per(S)\leq m/(256\dontcare)$.
Let $Q = S[1\dd \per(S)]$. By definition of $\per(S)$, $S$ is a prefix of $Q^\infty$ and by~\cref{fact:period_primitive}, $Q$ is primitive. 
Using \cref{fact:compute_mispers}, we compute the sets $\lmisp(P,x,y,1)$ and $\rmisp(P,x,y,1)$ in $\cO(\gaps)$ time in the \pillar model.
In other words, we compute the maximal fragment $V$ of $P$ that contains $S$ and matches exactly some substring of $Q^\infty$.

\subparagraph{Subcase (a): $V=P$.} We conclude that $q := |Q| \leq m/(256\dontcare)$ is a deterministic period of~$P$.
We replace $Q$ by its (possibly trivial) rotation $Q_0$ such that~$P$ is equal to a prefix of $Q_0^\infty$ 
and then apply \cref{fact:relevant-fragment} to compute, in $\cO(D)$ time in the \pillar model, a fragment $T'$ of~$T$ that contains the same number of occurrences of $P$ as $T$, is
at Hamming distance $\cO(D)$ from a prefix of $Q^\infty$, and only has occurrences of $P$ at positions equivalent to $1 \pmod{q}$.

\newcommand{\hidden}{\textsf{Hidden}}

It now suffices to show how to compute the occurrences of $P$ in $T'$.
As in previous works~\cite{DBLP:conf/soda/BringmannWK19,unified}, we take a sliding window approach. Let $W$ be the set of positions in $P$ where we have a wildcard.
For $i \in [1 \dd |T'|-m+1]$, define $\hidden(i)$ to be the size of the intersection of $\Misp(T',Q) \cap [i \dd i+m)$ with $i + W$.\footnote{For a set $Y$ and an integer $z$, by $z+Y$ we denote the set $\{z+y : y \in Y\}$.}
Intuitively, this is the number of mismatches between $T'[i \dd i+m)$ and $Q^\infty$ that are aligned with a wildcard in $P$ (and are hence ``hidden'')
when we align $P$ with $T'[i \dd i+m)$.
$\hidden(\cdot)$ is a step function whose value changes $\cO(\dontcare \gaps)$ times as we increase $i$,
since each mismatch enters or exits the window $[i \dd i+m)$ at most once and whether it is hidden or not changes at most $2\gaps$ times.
We compute $\hidden(1)$ and store the positions where the function changes (as well as by how much)
as events in the increasing order; this sorting takes $\cO(DG\log\log D)$ time~\cite{Han2004}.

  For a position $i \leq |T'|-m+1$ with $i \equiv 1 \pmod{q}$, we have
  \[d_i := \Ham(T'[i\dd i+m),P) = \Misp(T'[i\dd i+m), Q) - \hidden(i).\]
  We maintain this value as we, intuitively, slide $P$ along $T$, $q$ positions at a time.
  If there are no events in $(i \dd i+q] \subseteq [1\dd |T'|]$, then $d_i = d_{i+q}$.
  This allows us to report all occurrences of $P$ in $T$ efficiently as $\cO(\dontcare \gaps)$ arithmetic progressions with
  common difference $q$ by processing all events in a left-to-right manner in $\cO(\dontcare \gaps)$ time.

\subparagraph{Subcase (b): $V \neq P$.} 
Our goal is to show that, in this case, the occurrences of $P$ in $T$ are $\cO(\dontcare)$ and they can be computed in time $\cO(\gaps\dontcare)$.
Without loss of generality, assume that $V$ is not a prefix of $P$.
This means that $\lmisp(P,x,y,1) = \{\mu\} \neq \{0\}$.
The occurrences of $S$ in $T$ give us a collection $\mathcal{S}$ of $\cO(\gaps)$ $S$-runs in $T$, any two of which can overlap by less than $\per(S)=q$ positions due to \cref{fact:sruns}.
For each $S$-run $R$, extend $R$ to the left until either of the following two conditions is satisfied:
\begin{enumerate}[(a)]
\item the ratio of encountered misperiods to the sum of $|R|$ and the number of prepended positions exceeds $20 \dontcare/m$,
\item the beginning of $T$ has been reached.
\end{enumerate}
Denote by $E_R$ the resulting fragment of $T$
and by $\mathcal{M}_R$ the set of misperiods in it.
The following two claims are of crucial importance for the algorithm's performance. The proof of the first one can be found in \cref{sec:mispers_comp} along with an illustration.

\begin{restatable}{claim}{claimalign}\label{claim:align}
If $p+1$ is an occurrence of $P$ in $T$ that aligns $S$ with an occurrence of $S$ in an $S$-run $R=T[r \dd r']$, then $p+\mu \in \M_R$.
\end{restatable}

\begin{claimproof}
We first show that $p+\mu \in \lmisp(T,r,r',|T|)$.
We have that the solid character $P[\mu]$ is different from a character $P[\pi]$, where $\pi \in [x \dd y]$ and $\pi-\mu \equiv 0 \pmod q$.
Further, $P[\pi] = T[p+\pi]$, where $p+\pi \in [r \dd r']$, and $P[\mu] = T[p+\mu]$.
This means that $T[p+\mu] = P[\mu] \neq P[\pi] = T[p+\pi]$.
Now, since $(p + \pi) -  (p + \mu) = \pi - \mu$ is divisible by~$q$, which is the period of $T[r \dd r']$, we have $p+\mu \in \lmisp(T,r,r',|T|)$; see~\cref{fig:run_misperiods}.

\begin{figure}[htpb!]
\begin{center}
\begin{tikzpicture}[scale=0.6]

\draw[thick] (0,0) rectangle (21,1);
\node[left] at (0,0.5) {$T$};
\node[above] at (2,-0.7) {\small{$p+1$}};
\node[above] at (4,-0.7) {\small{$p+\mu$}};
\node[above] at (6,-0.7) {\small{$p+\nu$}};
\node[above] at (8,-0.65) {\small{$r$}};
\node[above] at (16,-0.7) {\small{$p+\pi$}};
\node[above] at (20,-0.65) {\small{$r'$}};

\draw[pattern=north west lines, pattern color=red] (8,0) rectangle (20,1) node[pos=0.5] {\small{$R$}};
\draw[pattern=north east lines, pattern color=green] (6,0) rectangle (8,1);
\draw[<->] (6,-1.25)--(20,-1.25) node[pos=0.5,above] {\small{$E_R$}};

\draw (8,1) to[out=90,in=90] node[midway,below] {\scriptsize{$\per(S)$}}  (11,1);
\draw (11,1) to[out=90,in=90]  (14,1);
\draw (14,1) to[out=90,in=90]  (17,1);
\draw (17,1) to[out=90,in=90] (20,1);

\draw (8,2.25) -- (14,2.25);
\draw[thick] (11,2.5) -- (17,2.5);
\draw (14,2.75) -- (20,2.75);
\node[left] at (8,2.25) {\scriptsize{$S$}};
\node[left] at (11,2.5) {\scriptsize{$S$}};
\node[left] at (14,2.75) {\scriptsize{$S$}};
 
\draw[thick] (2,-4) rectangle (18,-3);
\node[left] at (2,-3.5) {$P$}; 
\node[above] at (4,-4.7) {\small{$\mu$}};
\node[above] at (6,-4.6) {\small{$\nu$}};
\node[above] at (11,-4.65) {\small{$x$}};
\node[above] at (16,-4.65) {\small{$\pi$}};
\node[above] at (17,-4.75) {\small{$y$}};

\draw[pattern=dots, pattern color=gray!40] (11,-4) rectangle (17,-3) node[pos=0.5] {\small{$S$}};

\draw (4,-2.2) to[out=-15,in=90]  (5,-3);
\draw (5,-3) to[out=90,in=90]  (8,-3);
\draw (8,-3) to[out=90,in=90]  (11,-3);
\draw (11,-3) to[out=90,in=90] node[midway,below] {\scriptsize{$\per(S)$}} (14,-3);
\draw (14,-3) to[out=90,in=90]  (17,-3);

\draw[dashed] (2,-2.8)--(2,-0.6);
\draw[dashed] (4,-2.8)--(4,-0.6);
\draw[dashed] (6,-2.8)--(6,-0.6);
\draw[dashed] (16,-2.8)--(16,-0.6);

\draw[<->] (6,-5)--(17,-5);
\node at (10.5,-5.4) {\small{no misperiods}};
\node at (10.5,-6) {\small{$\le |E_R| \cdot 16D/m$ wildcards}};

\end{tikzpicture}
\end{center}
\caption{The run $R$ and an occurrence $p$ of $P$ in $T$ that aligns $S$ with an occurrence of $S$ in $R$.}
\label{fig:run_misperiods}
\end{figure}
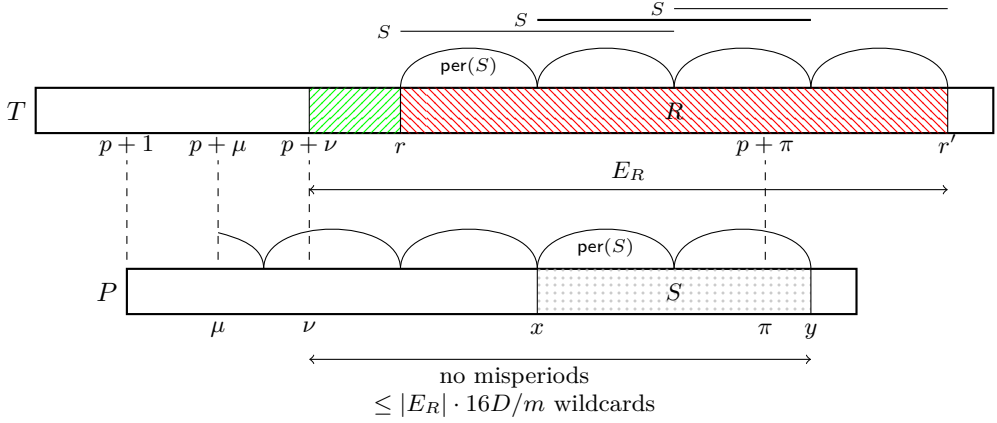

Now, assume for sake of a contradiction that $p+\mu \not\in \M_R$. Intuitively, this can only be the case if the extension of $R$ did not reach the beginning of $T$ due to encountering too many misperiods. On the other hand, the fragment $P[\mu \dd x) S$ of $P$ contains only one misperiod and cannot contain many wildcards, since every position of $S$ is a sparsifier. As a result, a misperiod in $T$ will be aligned with a position of $P$ that is neither a misperiod nor a wildcard, contradicting the fact that $p+1$ is an occurrence of $P$ in $T$. Formally, let $p+\nu > p+\mu$ be the misperiod that forced the extension algorithm to stop, i.e. $E_R = T[p+\nu \dd r']$. (See~\cref{fig:run_misperiods}.) 
Then, the stopping condition implies that $|\M_R| > |E_R| \cdot 20\dontcare/m$. All these misperiods belong to the prefix of $E_R$ matching $P[\nu \dd y]$. On the other hand, due to $\mu < \nu$, we have $[\nu \dd y] \cap \lmisp(P,x,y,|P|) = \emptyset$. 
Furthermore, every position of $S$ is a sparsifier, and therefore the number of wildcards in $P[\nu \dd y]$ is at most $|E_R| \cdot 16D/m$.
Thus, there exists $p+s \in \M_R$ such that $s \not\in \lmisp(P,x,y,|P|)$ and $P[s]$ is not a wildcard.
This implies that $T[p+s]\neq P[s]$, a contradiction to the fact that $p+1$ is an occurrence of $P$. \qedhere
\end{claimproof}

\begin{claim}
  The set $\cup_{R \in \mathcal{S}} \M_R$ is of size $\cO(D)$
  and it can be computed using $\cO(D)$ \pillar operations given the set $\S$ of $S$-runs and $q$.
\end{claim}
\begin{claimproof}
We start by initialising a set $\R := \S$, marking every element of $\R$ as unprocessed and a set $\M = \emptyset$.
We then iteratively perform the following procedure for the rightmost unprocessed $R = T[r\dd r'] \in \R$.
Compute $\mathcal{M}_R$ using \cref{lem:compute_mispers}, set $\mathcal{M}:=\mathcal{M} \cup \mathcal{M}_R$, and mark $R$ as processed.
This takes time proportional to the sum of $|\M_R|$ and the number of groups of wildcards contained in $E_R$.
Let us say that two elements $R = T[r\dd r']$ and $R'=T[t\dd t']$ of $\mathcal{S}$ are \emph{synchronised} if and only if $r = t \pmod{q}$.
During the procedure, whenever we compute some $E_R = T[x \dd r']$ that extends beyond an (unprocessed) $S$-run $R' = T[t \dd t']$, that is, $x \leq t \leq t' < r'$, and $R$ and $R'$ are synchronised, we remove $R'$ from $\R$---the total time required for this step is $\cO(G)$.

We now show the correctness of the algorithm.
If, while extending a run $R = T[r \dd r'] \in \R$, we extend beyond a run $R'=T[t \dd t'] \in \R$ with $r = t \pmod{q}$, observe that the left endpoint of $E_R$ cannot be to the right of the left endpoint of $E_{R'}$, since we have at least as big a budget for misperiods in the extension of $R$ when we reach position $t$ as in the extension of $R'$ when we reach position $t$.
This implies that $\M_{R'} \subseteq \M_R$ and hence the algorithm correctly computes $\M = \cup_{R \in \mathcal{S}} \M_R$.
Additionally, it guarantees that any computed $E_R$ and $E_{R'}$ for synchronised $S$-runs $R$ and $R'$ are disjoint.

Finally, we analyse the algorithm's time complexity. Henceforth, $\R$ denotes set of runs that were processed.
Observe that the run extensions take $\cO(\sum_{R \in \R} |\M_R|)$ time in total in the \pillar model.
As we have
\[\sum_{R \in \R} |\M_R| \leq |\R| + \sum_{R \in \R} |E_R| \cdot 20\dontcare/m \leq \cO(G) + 20\dontcare/m \cdot \sum_{R \in \R} |E_R|,\]
proving that $\sum_{R \in \R} |E_R| = \cO(m)$ directly yields that $\M = \cO(D)$ and that the algorithm takes $\cO(\dontcare)$ time.

In what follows, we ignore all $E_R$ that are of length at most $m/D$ as their total length is $\cO(G \cdot m/D)=\cO(m)$.
Let us partition~$T$ into a collection $\mathcal{Q}=\{T[1+iq \dd (i+1) q] : i \in[0 \dd \lfloor n/q \rfloor-1]\}$ of consecutive fragments of length~$q$, with the last one potentially being shorter and in this case discarded.
We say that an element $T[i \dd j]$ of $\mathcal{Q}$ is \emph{synchronised} with an element $R=T[r\dd r']$ of $\mathcal{R}$ if no position in $[i \dd j]$ is a misperiod with respect to $T[r\dd r']$.
For a run $R=T[r\dd r']$, let $\mathcal{Q}_R = \{T[i \dd j] \in \mathcal{Q} : r\leq i\leq j\leq r'\}$ consist of all elements of $\mathcal{Q}$ that are fully contained in $E_R$ and observe that \[|\mathcal{Q}_R| \geq |E_R|/q - 2 \geq |E_R|/(m/256D) - 2 = |E_R|\cdot 256D/m - 2\geq |E_R|\cdot 252D/m+2.\]
Further, let  $\mathcal{Q}_R^s = \{X \in \mathcal{Q}_R : X \text{ is synchronised with } R\}$.
As $E_R$ contains at most $|E_R|\cdot 20D/m+1$ misperiods with respect to $T[r\dd r']$, we have $|\mathcal{Q}_R^s|\geq |\mathcal{Q}_R|/2$.
This means that $|E_R| = \cO(|\mathcal{Q}_R^s|\cdot q)$.
Now, observe that if some element of $\mathcal{Q}$ is synchronised with two elements $R$ and $R'$ of $\mathcal{R}$, then $R$ and $R'$ are themselves synchronised.
Since the computed extensions of synchronised runs are pairwise disjoint, the considered sets~$\mathcal{Q}_R^s$ are pairwise disjoint and hence the bound follows:
\[\sum_{R \in \R} |E_R| = \cO(m) + \sum_{R \in \R, |E_R|\geq m/D} |E_R| = \cO(m + |\mathcal{Q}|\cdot q) = \cO(m).\]
\end{claimproof}

We can now conclude the proof of the theorem. By \cref{claim:align}, the starting positions of occurrences of $P$ in $T$ are in the set $\{\nu - \mu + 1 : \nu \in \M\}$.
This concludes the proof of the combinatorial bound, as the size of this set is $\cO(\dontcare)$.
As for the time complexity, we verify each candidate position using \cref{fact:kangaroo} in
total time $\cO(\dontcare \gaps)$ in the \pillar model.
\end{proof}

\section{Pattern Matching with \texorpdfstring{$k$}{k} Mismatches in the \pillar Model}\label{sec:ham-pm}

In this section, we extend the results of \cref{sec:exact-pm}, showing that
the \textit{$k$-mismatch occurrences} of a pattern $P \in \Sigma_{\wild}^*$ in a solid text $T$ can be computed in
$\cOtilde((\dontcare +\gaps)\cdot (\gaps + k))$ time in the \pillar model.
Further, we prove that the starting positions of these occurrences can be decomposed
into $\cO((\dontcare + k)\gaps)$ arithmetic progressions
with the same difference,
plus $\cO((\dontcare + k)k)$ additional $k$-mismatch occurrences.

\hampm*

\subsection{Computing Structure in the Pattern}\label{sec:decomposition}
We start by showing a \emph{decomposition lemma},
that either extracts useful structure
from the pattern or reveals that it is close to a periodic string.
Our lemma is analogous to the decomposition lemma for the case where both strings are solid~\cite[Lemma 3.6]{unified}.
The crucial differences are two:
\begin{itemize}
\item in Case~\ref{case:breaks}, we require that breaks are solid strings,
\item in Case~\ref{case:rep-regions}, we ensure that each computed repetitive region contains a sparsifier.
\end{itemize}

\begin{restatable}{lemma}{hammingDecomp}\label{lemma:ham-decomp}
Let $P$ be a string of length $m$ that contains $\dontcare\le m/16$
wildcards arranged in~$\gaps$ groups.
Further, let $k \in [1\dd m]$ be an integer threshold,
  and let $\gamma := \gaps + k$ and $\tau := \dontcare + k$.
  At least one of the following holds:
  \begin{enumerate}[(I)]
    \item\label{case:breaks} $P$ contains $2\gamma$ disjoint
    solid strings $B_1,\ldots,B_{2\gamma}$, that we call \emph{breaks},
    each having length $m/(16\gamma)$ and the period greater than $m/(512\tau)$.
    \item\label{case:rep-regions} $P$ contains $r$ disjoint
    \emph{repetitive regions} $R_1,\ldots, R_r$ of total
    length $m_R \ge m/8$, such that, for every $i$:
    \begin{itemize}
    \item $R_i$ contains a sparsifier,
    \item $|R_i| \ge m/16\gamma$, and,
    \item for a
    primitive string $Q_i$ with $|Q_i|\le m/(512\tau)$,
    we have $\Ham(R_i, Q_i^\infty) = \ceil{32k/m\cdot|R_i|}$.
    \end{itemize}
    \item\label{case:periodic-p} There exists a primitive string $Q$ of length at most $m/(512\tau)$
    such that $\Ham(P, Q^\infty) \le 32k$.
  \end{enumerate}
  Moreover, there is an algorithm that takes $\cO(G+k)$ time in the
  \pillar model and distinguishes between the above cases, returning one of the following:
  either $2\gamma$ disjoint breaks,
  or repetitive regions $R_1,\ldots, R_r$ of total length at least $m/8$ along with primitive strings $Q_1,\ldots, Q_r$,
  or a primitive string~$Q$ along with $\Misp(P,Q)$.
\end{restatable}
\begin{proof}
We process $P$ from left to right. 
While we have not yet arrived to one of the Cases~\ref{case:breaks}-\ref{case:periodic-p}, we repeat the following procedure.
We take the leftmost fragment $F$ of $P$ of length $m/(16\gamma)$ that starts to the right of the current position $j$ and consists only of sparsfiers.
If the period of $F$ is greater than $m/(512\tau)$, then we add $F$ to the set of breaks and proceed to position $j+|F|$.
Otherwise, we extend $F$ to the right until the either number of mismatches between $F$ and $Q^\infty$, where $Q := F[1\dd \per(F)]$, becomes equal to $\lceil{32k/m\cdot|F|\rceil}$ or we reach the end of $P$.
In the former case, we add $F$ to the set of repetitive regions and proceed to position $j+|F|$.
In the latter case, we extend $F$ to the left until either the number of mismatches between $F$ and $Q'^\infty$, where $Q' = \rot^{|F|-m+j} (Q)$, becomes equal to $\lceil{32k/m\cdot|F|\rceil}$, or we reach the start of $P$.
If we accumulate enough mismatches, we update the set of repetitive regions to be $\{F\}$ and terminate the algorithm.
Otherwise, i.e., if we reach the start of $P$, we conclude that $P$ is at Hamming distance at most $32k$ from a solid string with period at most $m/(512\tau)$, and we are hence in Case~\ref{case:periodic-p}.

Let us now prove the correctness of the above algorithm.
We first show that if the algorithm has not already concluded that we are in one of the three cases, then there exists a fragment $F$ of sparsifiers that starts in $[j \dd 7m/8]$, where $j$ is our current position.
A direct application of \cref{lem:toy_problem_better} implies that there are at least $m/2 - \dontcare - m/8 \ge 5m/16$ sparsifiers in $[1 \dd 7m/8]$ (since $\dontcare \le m/16$), and they are arranged in $\gaps + 1$ intervals.
By construction, the sparsifiers in an interval are covered from left to right, and hence at most $m/16\gamma-1$ positions can be left uncovered in each interval.
Under our assumptions, the breaks and repetitive regions that have been already computed cover less than $2\gamma \cdot m/(16\gamma) + m/8 = m/4$ positions.
If there is no fragment~$F$ of length $m/16\gamma$ that consists only of sparsifiers and starts in $[j \dd 7m/8]$, then at least
\[5m/16-(m/16\gamma-1)\cdot (\gaps+1)>5m/16-(m/16\gamma)\cdot (\gaps+1)\ge m/4\]
of the sparsifiers have been already covered, a contradiction.
From the above, it also follows that if a fragment $F$ reaches the end of the pattern during its extension, then its length becomes at least $m/8$.
Then, if such a fragment is extended to the left and accumulates enough mismatches with respect to the periodicity before the start of $P$ is reached,
we can set our set of repetitive regions to $\{F\}$ and terminate the algorithm. This completes the proof of the correctness of the structural result.

Now, note that the presented proof is algorithmic. \cref{lem:toy_problem_better} computes the set of sparsifiers, represented as $\cO(G)$ disjoint intervals, in $\cO(G)$ time. After this preprocessing, the procedure can retrieve a new fragment $F$ in constant time. In total, it considers $\cO(G)$ fragments.
  Computing the period of a solid string takes $\cO(1)$ time in the \pillar model~\cite{unified,DBLP:conf/soda/KociumakaRRW15}.
  Moreover, in our attempt to accumulate $\cO(k)$ misperiods in total, we encounter each group of
  wildcards at most twice: at most once when extending to the right and at most once when extending to the left.
  All such extensions thus take $\cO(G+k)$ time in the \pillar model due to \cref{lem:compute_mispers}.
\end{proof}

\subsection{The Almost Periodic Case}\label{sec:ham-pm-periodic}
Case~\ref{case:periodic-p} is treated quite similarly to Case~\rom{2}(a) of the exact pattern matching algorithm (see \cref{sec:exact-pm}).
 
\begin{restatable}{lemma}{periodicOccs}\label{lemma:per-region-occs}
  Let $S$ be a pattern of length $m$ with $\dontcare$ wildcards
  arranged in $\gaps$ groups\footnote{In the final algorithm, $S$ is a fragment of the pattern $P$, potentially much shorter than the text.},
  let~$T$ be a solid text of length $n$, let $k \in [0 \dd m]$,
  and let $d \ge 2(k+\dontcare)$ be a positive integer.
  If there exists a primitive string $Q$ with $|Q|\le m / 8d$
  such that $\Ham(S, Q^\infty) \le \min\{d,32k\}$, then
  we can compute a representation of $\occhk(S, T)$  as $\cO(d(\gaps+k))$ arithmetic progressions with common difference $|Q|$
  in $\cO(d(\gaps + k)\log\log d \cdot n/m)$ time
  plus $\cO(d \cdot n/m)$
  \pillar operations.
  Moreover, if $\Ham(S, Q^\infty) \ge 2k$,
  then $|\occhk(S, T)| = \cO(d \cdot n/m)$.
\end{restatable}
\begin{proof}
  We only consider the case when $n \le 3m/2$.
  The result then follows using the ``standard trick'' presented in the introduction.

  We use an event-driven scheme that extends the one used in the almost periodic case of \cref{sec:exact-pm}. First, we apply \cref{fact:relevant-fragment} to compute a fragment $T'$ of $T$ that contains
  the same number of $k$-mismatch occurrences as $T$, is at Hamming distance $\cO(d)$ from a prefix of $Q^\infty$,
  and only has occurrences of $P$ at positions that are equivalent to $1 \pmod{|Q|}$.
  This takes $\cO(d)$ time in the \pillar model.
  We then apply \cref{fact:infty} to compute $\Misp(T',Q)$ in $\cO(d)$ time in the \pillar model.
  For a position $i \le |T'|-m+1$, the distance $d_i$ between $T'[i\dd i+m)$
  and $S$ is given by
  \[ d_i = |\Misp(S, Q)| + |\Misp(T'[i\dd i+m), Q)| 
        - 2\Matching(i) - \Aligned(i)  - \Hidden(i), \text{ where}\] 
  \begin{itemize}
    \item $\Matching(i)$ is the number of positions that are mismatches between $S$ and $Q^\infty$, and $Q^\infty$ and $T'$, but not between $S$ and $T$, i.e.,
    \[\Matching(i) = |\{j : j\in \Misp(S, Q) \cap \Misp(T'[i\dd i+m), Q) \land S[j] = T'[i+j] \land S[j] \neq \wild\}|.\]
    \item $\Aligned(i)$ is the number of positions that are mismatches between $S$ and $Q^\infty$, $Q^\infty$ and $T'$, and $S$ and $T$ (as opposed to $\Matching(i)$),
    i.e.,
    \[\Aligned(i) = |\{j : j\in \Misp(S, Q) \cap \Misp(T'[i\dd i+m), Q) \land S[j] \neq T'[i+j] \land S[j] \neq \wild\}|.\]
    \item $\Hidden(i)$ is the number of positions that are mismatches between $T'$ and $Q^\infty$, and that are aligned with
    wildcards, i.e.,
    \[\Hidden(i) = |\{j : j\in \Misp(T'[i\dd i+m), Q) \land S[j] = \wild\}|.\]
  \end{itemize}
  Recall that every $j\in \occhk(S, T')$ satisfies $j \equiv 1 \pmod{|Q|}$ (\cref{fact:thm-unified}),
  hence we only consider the values of $d_i$ as $i$ increases by multiples
  of $|Q|$. 
  Then, the value $d_i$ only changes when one of the following events occurs:
  a position in $\Misp(T',Q)$ enters or exits the active window $T'[i\dd i+m)$,
  starts or stops being aligned with a group of wildcards,
  or starts or stops being aligned with a position in~$\Misp(S,Q)$.
  As $|\Misp(T',Q)| = \cO(d)$, there are~$\gaps$ groups of wildcards in $S$ and $|\Misp(S,Q)| = \cO(k)$, there are $\cO(d(\gaps+k))$
  events. 

  If $d_i \le k$, then all positions equivalent to $1 \pmod{|Q|}$ until the subsequent event
  are $k$-mismatch occurrences, and form an arithmetic
  progression with difference $|Q|$.
  As there are $\cO(d(\gaps+k))$ events, 
  we obtain the stated bound on the number of arithmetic progressions.

  We sort the events by index in $\cO(d(\gaps+k)\log \log d)$ time~\cite{Han2004}.
  Then, we process them from left to right; processing
  one event takes constant time.
  The initial value of $d_0$ can be computed in time linear in the number
  of events.
  Overall, the running time is dominated by the sorting operation.
  We additionally perform $\cO(d)$ \pillar operations to compute $T'$.

  Finally, if $\Ham(S, Q^\infty) \ge 2k$, then in any $k$-mismatch occurrence
  of $S$ in $T$, at least $2k-k = k$ positions in $\Misp(S,Q)$ are aligned with positions in $\Misp(T',Q)$, as otherwise there would be more than $k$ mismatches.
  As we have $|\Misp(S,Q)| = \Ham(S, Q^\infty) \geq 2k$
   and $|\Misp(T',Q)| = \cO(d)$,
  there are $\cO(dk)$ ways of aligning misperiods.
  We thus have $|\occhk(S, T)| = \cO(dk/k) = \cO(d)$ $k$-mismatch occurrences.
\end{proof}

\subsection{The Remaining Cases}\label{sec:ham-pm-aper}

We now show that in each of the Cases~\ref{case:breaks} and~\ref{case:rep-regions}
of \cref{lemma:ham-decomp}, $T$ contains $\cO(\dontcare + k)$
$k$-mismatch occurrences of $P$,
and we can efficiently compute a set $\mathcal{S}$ of $\cO(\dontcare + k)$ positions
that contains all the starting positions of $k$-mismatch occurrences of $P$.
We then verify each of the candidate positions in $\mathcal{S}$ in $\cO(\gaps+k)$ time using \cref{fact:kangaroo}.
    
We first handle the case when the pattern contains $2\gamma$ disjoint breaks.
\begin{lemma}\label{lemma:case-breaks}
  In Case~\ref{case:breaks} of \cref{lemma:ham-decomp}, a solid text $T$ of length at most $3m/2$ contains $\cO(\dontcare + k)$ $k$-mismatch
  occurrences of $P$.
  Moreover, we can compute in $\cO(\dkgk)$ time in the \pillar model
  a set $\mathcal{S} \supseteq \occhk(P, T)$ of size $\cO(\dontcare + k)$.
\end{lemma}
\begin{proof}
  Let $\{B_i\}$ be the breaks computed by the algorithm of \cref{lemma:ham-decomp}. For every $i$, let $p_i$ denote the starting position of $B_i$ in $P$.
  For every exact occurrence $j$ of $B_i$ in $T$,
  we put a mark at position $j-p_i+1$ in $T$.
  As $B_i$ has period greater than $m/512\tau$, $T$ contains
  at most $768\tau$ occurrences of $B_i$,
  and they can be computed in $\cO(|T|/|B_i|+\tau) = \cO(\gamma + \tau) = \cO(\dontcare + k)$ time in the \pillar model.
  As there are $2\gamma$ breaks,
  we place at most $768\tau\cdot2\gamma = 1536\tau\gamma$ marks in total.

  Now, in every occurrence of $P$, at most $k$ of the $B_i$s are not matched exactly
  and hence at least $2\gamma-k$ of the $B_i$s are matched exactly.
  Thus, every position $j\in\occhk(P, T)$ has at least $2\gamma - k \ge \gamma$
  marks.
  We designate $\mathcal{S}$ to be the set of positions with at least $2\gamma-k$ marks.
  Observe that $\occhk(P, T) \subseteq \mathcal{S}$ and $|\mathcal{S}| \le 1536\tau\gamma / (2\gamma-k) \le 1536\tau = \cO(\dontcare + k)$. Thus, $\mathcal{S}$ satisfies the conditions of the lemma's statement.
\end{proof}

We obtain a similar result for repetitive regions via a more sophisticated marking scheme.

\begin{lemma}\label{lemma:case-rep-regions}
  In Case~\ref{case:rep-regions} of \cref{lemma:ham-decomp}, given a solid text $T$ of length at most $3m/2$,
  we can compute in $\cO(\dkgk \log\log(\dontcare + k))$ time plus $\cO(\dkgk)$ \pillar operations
  a set $\mathcal{S} \supseteq \occhk(P, T) $ of size $\cO(\dontcare + k)$.
\end{lemma}
\begin{proof}
Let $\{R_i\}$ be the repetitive regions computed by the algorithm of \cref{lemma:ham-decomp}. 
  For every~$i$, let $D_i$ denote the number of wildcards in $R_i$,
  $d_i = \lceil 32(k+\dontcare)/m \cdot |R_i|\rceil$, and $k_i = \lfloor 16k/m \cdot |R_i|\rfloor$.
  For every $i$ and every $k_i$-mismatch occurrence of $R_i$,
  we put a mark of weight $|R_i|$ at the corresponding starting position for $P$.
  
  \begin{claim}\label{claim:few-occs-and-marks}
    For every $i$, there are $\cO(\dontcare + k)$ $k_i$-mismatch occurrences of
    $R_i$ in $T$.
    Moreover, the total weight of marks is $\cO((\dontcare + k) \cdot m_R)$.
  \end{claim}
  \begin{claimproof}
    We apply \cref{lemma:per-region-occs} to each repetitive region.
    First, let us show that the conditions of the lemma are satisfied. Recall that $|R_i| \geq m/(16(G+k))$ and hence $16k/m \cdot |R_i| \geq 1$.
    \begin{itemize}
    \item $d_i \geq 2(k_i+D_i)$: As $R_i$ contains a sparsifier, we have $D_i \leq 16D/m \cdot |R_i|$ and hence the bound follows.
    \item $|Q_i| \leq |R_i|/8d_i$: We have $d_i \leq 64 |R_i|(k+\dontcare)/m \Longleftrightarrow m \leq 64 |R_i|(k+\dontcare)/d_i$ and hence
    \[|Q_i|\leq m/(512\tau) = m/(512(k+D))  \leq 64 |R_i|/(512d_i) = |R_i|/(8d_i).\]
    \item $\Ham(R_i, Q_i^\infty) \leq \min\{d_i,32 k_i\}$: This  follows from the fact that $\Ham(R_i, Q_i^\infty) = \ceil{32k/m\cdot|R_i|}$.
    \end{itemize}
    
    Now, since $\Ham(R_i, Q_i^\infty) \geq 2 k_i$, $T$ contains $\cO(d_i n/|R_i|)$ occurrences of $R_i$.
    Hence, the total weight of marks for a given $R_i$ is
    \[w_i 
      = \cO(d_i n/|R_i|) \cdot |R_i|
      = \cO((\dontcare + k)\cdot |R_i|),\]
    since $m=\Theta(n)$, and, summing over $i = 1,\ldots, r \le 2\gamma$, we get that the total weight of marks that we place is $\cO((\dontcare + k) \cdot m_R)$. 
  \end{claimproof}

  We next lower bound the number of marks placed at a $k$-mismatch occurrence of~$P$.
  
  \begin{claim}  
    The total weight of marks placed in any position $\ell \in \occhk(P, T)$ is at least $m_R - m/16$.
  \end{claim}
  \begin{claimproof}
    Consider a $k$-mismatch occurrence of $P$ at position $\ell$ of $T$.
    For every $i$, let $r_i$ be the starting position of $R_i$ in $P$,
    and let $k_i' = \Ham(R_i, T[\ell+r_i\dd \ell+r_i+|R_i|))$ be the Hamming
    distance between $R_i$ and the fragment of $T$ with which it is aligned.
    As $\ell \in \occhk(P, T)$ and the $R_i$s are disjoint,
    we have $\sum_i k_i' \le k$.
    Now, let $I := \{i : k_i' \le k_i\}$. We have
    \[
      \sum_{i\notin I} |R_i|
      = \sum_{i\notin I} \frac{16mk}{16mk} \cdot|R_i|
      = \frac{m}{16k}\sum_{i\notin I} \frac{16k}{m} \cdot|R_i|
      < \frac{m}{16k}\sum_{i\notin I} k_i' 
      \le \frac{m}{16k}\sum_{i=1}^r k_i' 
      \le \frac{m}{16}.
      \]
    The total weight of the marks placed at position $\ell$ in $T$ due to $i \in I$ is $|R_i|$, which amounts to a total weight of at least
    $m_R - \sum_{i\notin I} |R_i| \ge m_R - m/16$.
  \end{claimproof}

Therefore, we can choose $\mathcal{S}$ to contain all positions to which we have placed marks of total weight at least $m_R - m/16$. 
  Dividing the total weight of marks which is $\cO((\dontcare + k) \cdot m_R)$ by $m_R - m/16$ which is at least $m_R/2$ since $m_R \geq m/8$, we obtain $|\mathcal{S}| = \cO(\dontcare + k)$.

  Finding the $k_i$-mismatch occurrences of $R_i$
  using \cref{lemma:per-region-occs}
  takes $\cO(d_i(G_i + k_i)\log\log d_i \cdot n/|R_i|) = \cO((\dontcare + k)(G_i + k_i)\log\log (\dontcare + k))$ time
  plus $\cO(d_i \cdot n/|R_i|) = \cO(\dontcare + k)$ \pillar operations.
  As the $R_i$s are disjoint, the sum of the $G_i$s is $\cO(G)$.
  Further, the sum of the $k_i$s is $\cO(k)$.
  Thus, summing over all $i$, computing the occurrences of all $\cO(\gaps +k)$ $R_i$s takes
  $\cO(\dkgk \log\log(\dontcare + k))$ time plus $\cO(\dkgk)$ \pillar operations.
  
  Therefore, computing the set $\mathcal{S}$ of possible starting positions
  of $P$ in $T$ can be done in $\cO(\dkgk\log\log (\dontcare + k))$ time
  plus $\cO(\dkgk)$ \pillar operations.
\end{proof}

\cref{lemma:case-breaks} and \cref{lemma:case-rep-regions} are the last
pieces needed to prove the algorithmic part of \cref{thm:ham-pm}.

\subparagraph{Proof of the algorithmic part of \cref{thm:ham-pm}.}
  If $P$ contains $\dontcare > m/16$ wildcards,
  then we apply kangaroo jumping (\cref{fact:kangaroo})
  to check whether each of the $n \le 3m/2$ positions in $T$
  is a $k$-mismatch occurrence of $P$. 
  This requires $\cO(m (\gaps+k)) = \cO(\dkgk)$ time in the \pillar model.

  Otherwise, we have $\dontcare \le m/16$, and we can run the algorithm of \cref{lemma:ham-decomp},
  which takes $\cO(\gaps + k)$ time in the \pillar model.
  In Cases~\ref{case:breaks} and~\ref{case:rep-regions}, we use \cref{lemma:case-breaks} and \cref{lemma:case-rep-regions},
  respectively, to compute using $\cO(\dkgk \log\log(\dontcare + k))$ time plus $\cO(\dkgk)$ \pillar operations
  a set $\mathcal{S}$ of size $\cO(\dontcare + k)$ that contains all
  $k$-mismatch occurrences of $P$.
  We verify each of these positions using \cref{fact:kangaroo}, in total time
  $\cO(\dkgk)$ in the \pillar model.
  In Case~\ref{case:periodic-p} of \cref{lemma:ham-decomp}, $P$ satisfies the conditions required
  for $S$ in \cref{lemma:per-region-occs},
  and we can apply this lemma to compute $\occhk(P, T)$
  using $\cO(\dkgk \log\log (\dontcare + k))$ time plus $\cO(\dontcare + k)$
  \pillar operations.

  In total, the algorithm uses $\cO(\dkgk \log\log (\dontcare + k))$ time
  and $\cO(\dkgk)$ \pillar operations.
\qed

\subparagraph{Proof of the combinatorial part of \cref{thm:ham-pm}.}\label{sec:fine-grained}

Finally, we explain how we can refine the analysis of \cref{lemma:per-region-occs}
to obtain a more precise characterisation of the structure of $k$-mismatch occurrences.

\begin{lemma}
  The $k$-mismatch occurrences of $P$ in $T$ can be decomposed
  into $\cO((\dontcare + k)\gaps)$ arithmetic progressions with
  common difference $q$ and $\cO((\dontcare + k)k)$ additional occurrences.
\end{lemma}
\begin{proof}
We proceed similarly to the proof of \cref{lemma:per-region-occs}, with $S = P$. Let $T'$ be the fragment of $T$ computed by~\cref{fact:relevant-fragment} using $d=\Theta(D+k)$. We have $|\Misp(T',Q)| = \cO(\dontcare+k)$. Define $d'_i = |\Misp(P,Q)| + |\Misp(T'[i \dd i+m),Q)|-\Hidden(i)$, where $\Hidden(i) = |\{j : j \in \Misp(T'[i \dd i+m),Q) \land P[j] = \wild\}|$. Note that $\Ham(T'[i \dd i+m),P) = d'_i - 2\Matching(i) - \Aligned(i)$, where
 \begin{align*}
	\Matching(i) &= \{j : j \in \Misp(P,Q) \cap \Misp(T'[i \dd i+m), Q) \land \wild \neq P[j] = T'[i+j]\},\\
	\Aligned(i) &= \{j: j \in \Misp(P,Q) \cap \Misp(T'[i \dd i+m), Q) \land \wild \neq P[j] \neq T'[i+j]\}
 \end{align*}
Therefore, $\Ham(T'[i \dd i+m),P) \le d'_i$. The inequality is strict if at least one of $\Matching(i)$ or $\Aligned(i)$ is positive. In particular, every position where $d_i' \le k$ corresponds to a $k$-mismatch occurrence of $P$ in~$T$. Using the event-driven scheme, we compute the values $d'_i$ for all $i \equiv 1 \mod |Q|$. The value $d'_i$ only changes when a position in $\Misp(T',Q)$ enters or exits the active window $T'[i\dd i+m)$, or when a position in $\Misp(T',Q)$ starts or stops being aligned with a group of wildcards in $P$. Therefore, the values $d'_i$ change $\cO(\gaps \cdot (\dontcare + k))$ times. As we only consider positions $i \equiv 1 \mod |Q|$, the set of positions where $d_i' \le k$ forms $\cO(\gaps \cdot (\dontcare + k))$ arithmetic progressions with common difference $|Q|$.

This analysis might have missed the $k$-mismatch occurrences where at least one of $\Matching(i), \Aligned(i)$ is positive. However, this requires a misperiod of $P$ to be aligned with a misperiod of $T'$. There are $\cO(k)$ of the former and $\cO(\dontcare+k)$ of the latter, hence the number of such occurrences is $\cO((\dontcare + k)k)$.
\end{proof}

\section{A Lower Bound on the Number of Arithmetic Progressions}\label{sec:lower-bound}
In this section we show a lower bound on the number of arithmetic progressions covering the set of $k$-mismatch occurrences of a pattern in a text.

\begin{restatable}{theorem}{lowerbound}\label{thm:lower-bound}
  There exist a pattern $P$ of length $m = \Omega((\dontcare+k)^{1+o(1)}(k+1))$
  and a text~$T$ of length $n\le 3m/2$ such that 
  the set of $k$-mismatch occurrences of $P$ in~$T$ 
  cannot be covered with less than $\Omega((\dontcare+k) \cdot (k+1))$
  arithmetic progressions.
\end{restatable}
\begin{proof}
We call a set $S \subseteq [1 \dd n]$ progression-free if it contains no non-trivial arithmetic progression, that is, three distinct integers $a,b,c$ such that $a+b-2c=0$. 

\begin{fact}[\cite{DBLP:conf/soda/Elkin10}]\label{lemma:progression-free-sets}
    For any sufficiently large $M$, there exists
    an integer $n_M = \cO(M 2^{\sqrt{\log M}})$
    and a progression-free set $S$
    such that $S$ has cardinality $M$ and $S \subseteq [n_M]$.
\end{fact}

Let $M = \dontcare + k/2$ and $S \subseteq [n_M]$ be a progression-free set of cardinality $M$.
We encode~$S$ in a string $P_S$ of size $n_M$ as follows:
for every $i\notin S$ we set $P_S[i] = 0$,
and we arbitrarily assign $k/2$ $1$s and $\dontcare$ wildcards
to the remaining $\dontcare + k/2$ positions.
We then consider the pattern $P = 0^\ell P_S 0^\ell$,
where $\ell$ is a parameter to be determined later.
In what follows, let $m := 2\ell + n_M$ denote the length of $P$.

Now, let $M' = k/2 + 1$ and $S' \subseteq [n_{M'}]$ be a progression-free set of cardinality $M'$.
We set $T := 0^{m/2}B_1\ldots B_{n_{M'}}0^{m/2}$,
where $B_i = 0^{t-1}1$ if $i \in S'$, $B_i = 0^{t}$ otherwise, and $t = \lfloor m/(2n_{M'})\rfloor$.
We pick $\ell$ large enough such that $t \ge 10n_M$ and $2n_{M'}$ divides $m$.
We have $t = \lfloor m/(2n_{M'}) \rfloor = \lfloor(2\ell + n_M) / (2n_{M'})\rfloor \geq 10n_M$, which implies that it suffices for $\ell$
to be larger than $n_M n_{M'}$ by a constant factor, i.e., $m = \Omega((k+\dontcare)(k+1)2^{\sqrt{\log(k+D)}+\sqrt{\log(k+1)}})$. Observe that $i \in X$ if and only if there exists $j$ such that $P[j] \in \{1, \wild\}$ and $T[i+j] = 1$. Moreover, any pair of $1$s in $T$ are at least $t \ge 10n_M$ positions apart,
while the $1$s and wildcards of~$P$ all lie within an interval of size $n_M$. 
Therefore, for a given alignment of $P$ and $T$, there can be at most one
$1$ of $T$ that is aligned with a $1$ or a wildcard of $P$;
it follows that $\occhk(P,T)$  has cardinality $(\dontcare +k/2) \cdot (k/2+1) = \Omega((\dontcare + k) \cdot (k+1))$.

It remains to show that $\occhk(P,T)$ does not contain arithmetic progressions of length $3$. Assume for a sake of contradiction that there exist $x,y,z\in \occhk(P,T)$ with $x < y < z$ that form an arithmetic progression, i.e., $y-x = z-y$. 
     Let $i_x$ denote the index of the block $B_{i_x}$ of $T$ that contains
    the leftmost $1$ that is aligned with a $1$ or a wildcard of $P$: this $1$
    is at position $m/2 + i_x t$ in~$T$.
    Similarly, let $d_x$ be such that the corresponding aligned character is at
    position $\ell+d_x$ in $P$. Define $i_y,i_z, d_y, d_z$ similarly for $y,z$.
    We can express each $w\in \{x,y,z\}$ in terms of $i_w$ and $d_w$
    as $w = m/2 + i_w t - d_w - \ell + 1$.
    Combining the above equations for $x$, $y$, and $z$, we get $y-x = (i_y-i_x) t - (d_y-d_x)$ and 
        $z-y = (i_z-i_y) t - (d_z-d_y)$. By construction, $|d_y-d_x| \leq n_M$. As $t \ge 10 n_M$, the equality $y-x = z-y$ thus yields $i_y-i_x = i_z-i_y$ and $d_y-d_x = d_z-d_y$. However, as $x < y < z$, at least one of the above two equations involves
    non-zero values.
    In other words, there is an three-term arithmetic progression in either $S$ or~$S'$,
    contradicting the fact that they are progression-free.
\end{proof}

\bibliographystyle{plainurl}
\bibliography{references}
\clearpage

\appendix

\section{Verification of Occurrences via Kangaroo Jumping}\label{sec:kangaroo}

Here, we provide a proof of \cref{fact:kangaroo}, which we restate for convenience.

\kangaroo*
\begin{proof}
  It suffices to use the kangaroo jumping technique~\cite{kangaroo}.
  We start at position $1$ of the pattern with mismatch budget $k$.
  With one \lcp query, we reach the first position in the pattern where we have either a wildcard or a mismatch. If it is a wildcard, we jump to the next non-wildcard character using another \lcp query (in more detail, we apply \cref{fact:infty} with wildcards replaced by $\$ \notin \Sigma$); otherwise, we decrement the budget of mismatches by one and continue from the next position. If the budget of mismatches becomes zero before we reach the end of the $P$, then $P$ does not have a $k$-mismatch occurrence at position $p$; otherwise it does.
  In each iteration we either decrease the mismatch budget or jump over a group of wildcards.
  Hence, the total number of performed \lcp queries (which are the bottleneck) is~$\Oh(\gaps+k)$.
\end{proof}

\section{Omitted Proofs from \texorpdfstring{\cref{sec:exact-pm}}{Section 3}}\label{sec:mispers_comp}

Here, we provide proofs of \cref{lem:compute_mispers} and \cref{fact:sruns}, which we restate for convenience.

\computeMispers*
\begin{proof}
It suffices to describe how to compute elements of the set $\rmisp(V,i,j,|V|)$ as the computation of elements of the set $\lmisp(V,i,j,|V|)$ is symmetric.

We first compute the period $q$ of $V[i \dd j]$ in $\cO(1)$ time in the \pillar model. Let $Q := V[i \dd i+q)$ and observe that $Q^2$ is a prefix of $V[i \dd j]$ since the latter is periodic. Hence, in constant time, we can retrieve a fragment of $V$ equal to any desired rotation of $Q$.

Let us now discuss how to efficiently compute the first misperiod, that is, the first mismatch between $V[i \dd |V|]$ and $Q^\infty$.
Let $F$ be the maximal solid fragment that contains $V[i \dd j]$.
Using \cref{fact:infty}, in $\cO(1)$ time in the \pillar model, we either compute the first misperiod or reach a group of wildcards  and conclude that there are no misperiods in $F$.
In the latter case we do the following.
While we have not yet found a misperiod or reached the end of $V$, we consider the subsequent maximal solid fragment $Y$; we find the starting position $y$ of this fragment by applying \cref{fact:infty} to compute the length of the encountered group of wildcards.
The elements of $\rmisp(V,i,j,|V|)$ spanned by $Y$ are the mismatches of $Y$ and $Q^\infty[y-i \dd y-i+|Y|)$.
We can compute the smallest such element, if one exists using \cref{fact:infty}, or reach the next group of wildcards.

Subsequent misperiods are computed in an analogous manner using the kangaroo method: we try to compute them in maximal solid fragments using \cref{fact:infty} in $\cO(1)$ time in the \pillar model, while skipping any encountered groups of wildcards in $\cO(1)$ time in the \pillar model.
\end{proof}

\sruns*
\begin{proof}
Towards a contradiction suppose that we have two $S$-runs $R_1$ and $R_2$ that overlap by at least $\per(S)$ positions.
Since $S$-runs correspond to inclusion-maximal arithmetic progressions of occurrences of $S$, the difference of the starting positions
of $R_1$ and $R_2$ is not a multiple of $\per(S)$.
Without loss of generality, suppose that the starting position of $R_1$ is to the left of the starting position of $R_2$.
Then, $S[1\dd \per(S)] = R_2[1\dd \per(S)]$ matches a non-trivial rotation of itself, and is hence not primitive, which contradicts~\cref{fact:period_primitive}.
\end{proof}

\section{Fast Algorithms in Various Settings}\label{sec:pillar}

In this section, we combine the algorithm encapsulated in \cref{thm:ham-pm} with efficient implementations of the \pillar model in the standard, dynamic, fully compressed, and quantum settings,
thus obtaining efficient algorithms for (approximate) pattern matching with wildcards in these settings.
For ease of presentation, we use the terms ``exact occurrences'' and ``$0$-mismatch occurrences'' interchangeably.

\paragraph*{Standard setting}

In the standard setting, where a collection of strings of total length $N$ is given explicitly, \pillar operations can be performed in constant time after an $\cO(N)$-time preprocessing, cf.~\cite{unified}. We thus obtain the following result, by noticing that the $\log\log(D+k)$ factor in the complexity of \cref{thm:ham-pm} only comes from sorting subsets of $[1\dd n]$ of total size $\cO((\dontcare + k)(\gaps + k))$ in the calls to \cref{lemma:per-region-occs}, which we can instead do naively in $\cO(n+(\dontcare + k)(\gaps + k))$ time as we can batch the computations.

\begin{theorem}\label{thm:standard}
  Let $P$ be a pattern of length $m$ with $\dontcare$ wildcards arranged
  in $\gaps$ groups, $T$ be a solid text of length $n$, and $k \geq 0$ be an integer.
  We can compute a representation of the $k$-mismatch occurrences of $P$
  in $T$
  in $\cO(n+(\dontcare + k)\cdot (\gaps + k))$
  time.
\end{theorem}

\paragraph*{Fully compressed setting}

\newcommand{\gen}{\mathit{gen}}

For our purposes, a straight-line program (SLP) is a context-free grammar $\Gamma$ that consists of the set $\Sigma_{\wild}$ of terminals and a set $N_\Gamma = \{A_1,\dots,A_n\}$ of non-terminals such that each $A_i \in N_\Gamma$ is associated with a unique production rule
$A_i \rightarrow f_\Gamma(A_i) \in (\Sigma_{\wild} \cup \{A_j : j < i\})^*$. We can assume without loss of generality that each production rule is of the form $A \rightarrow BC$ for some symbols $B$ and~$C$ (that is, the given SLP is in Chomsky normal form).
Every symbol $A \in S_\Gamma:=N_\Gamma \cup\Sigma_{\wild}$ generates a unique string, which we denote by $\gen(A) \in \Sigma_{\wild}^*$. The string $\gen(A)$ can be obtained from $A$ by repeatedly replacing each non-terminal with its production.
We say that~$\Gamma$ generates $\gen(\Gamma) := \gen(A_n)$.

In the fully compressed setting, given a collection of straight-line programs (SLPs) of total size $n$ generating strings
of total length $N$, each \pillar operation can be performed in $\Oh(\log^2 N \log \log N)$ time after an $\Oh(n \log N)$-time preprocessing, cf.~\cite{unified}.
Additionally, using a dynamic programming approach, we can compute a linked-list representation of all wildcards in $P$ in $\cO(m+D)$ time.
If we applied \cref{thm:ham-pm} directly in the fully compressed setting, we would obtain $\Omega(N/M)$ time, where~$N$ and $M$ are the uncompressed lengths of the text and the pattern, respectively. Instead, we can adapt a dynamic programming approach described in~\cite[Section 7.2]{unified} to obtain the following result.

\begin{theorem}[Fully Compressed Setting]
Let $\Gamma_T$ denote a straight-line program of size~$n$ generating a solid string $T$, let $\Gamma_P$ denote a
straight-line program of size $m$ generating a string $P$ with $D$ wildcards arranged in $G$ groups, let $k \geq 0$ denote an integer threshold, and set $N := |T|$ and $M := |P|$.
We can compute the number of $k$-mismatch occurrences of $P$ in~$T$ in $\Oh(m \log N + n (D+k)(G+k) \log^2 N \log \log N)$ time.
All occurrences can be returned in extra time proportional to the output size.
\end{theorem}

\paragraph*{Dynamic setting}

Let $\mathcal{X}$ be a growing collection of non-empty persistent strings; it is initially empty, and then undergoes updates by means of the following operations:
\begin{itemize}
    \item $\texttt{Makestring}(U)$: Insert a non-empty string $U$ to $\mathcal{X}$
    \item $\texttt{Concat}(U,V)$: Insert string $UV$ to $\mathcal{X}$, for $U,V\in \mathcal{X}$
    \item $\texttt{Split}(U,i)$: Insert $U[0\dd i)$ and $U[i\dd |U|)$ to $\mathcal{X}$, for $U\in\mathcal{X}$ and $i\in[0\dd |U|)$.
\end{itemize}

By $N$ we denote an upper bound on the total length of all strings in $\mathcal{X}$ throughout all updates executed by an algorithm.
A collection $\mathcal{X}$ of non-empty persistent strings of total length $N$ can be dynamically maintained with operations $\mathtt{Makestring}(U)$,  $\mathtt{Concat}(U,V)$, $\mathtt{Split}(U,i)$ requiring time $\Oh(|U| \cdot \log N)$, $\Oh(\log^2 N)$ and $\Oh(\log^2 N)$, respectively, so that \pillar operations can be performed in time $\Oh(\log^2 N)$. Furthermore, one can find all occurrences of a pattern of length $\ell$ in a string in $\mathcal{X}$ in $\Oh(\ell + \log^2 N + \mathrm{occ} \cdot \log N)$ time, where $\mathrm{occ}$ is the size of the output. 
All stated time complexities hold with probability $1-1/N^{\Omega(1)}$; see \cite{gawrychowski2018optimal,unified}. 
To compute the linked list representation of the endpoints of groups of wildcards for a string $X \in \mathcal{X}$, we simply use pattern matching for a pattern $P = \#$. The complexity of this step is dominated by the time required by our (approximate) pattern matching algorithm.

\begin{theorem}[Dynamic Setting]\label{thm:dynamic}
A collection $\mathcal{X}$ of non-empty persistent strings of total length $N$ over alphabet $\Sigma_{\wild}$ can be dynamically maintained with operations $\mathtt{Makestring}(U)$, $\mathtt{Concat}(U,V)$, $\mathtt{Split}(U,i)$ requiring time $\Oh(|U| \cdot \log N)$, $\Oh(\log^2 N)$, and $\Oh(\log^2 N)$, respectively, so that, given two strings $P,T\in \mathcal{X}$, such that $P$ has $\dontcare$ wildcards arranged in $\gaps$ groups and~$T$ is a solid string, and an integer threshold $k \geq 0$,
we can return a representation of all $k$-mismatch occurrences of $P$ in $T$ in time $\cO((\dontcare+k)(\gaps+k) \cdot |T|/|P| \cdot \log^2 N)$ time.
All stated time complexities hold with probability $1-1/N^{\Omega(1)}$.
\end{theorem}

Kempa and Kociumaka~\cite[Section 8 in the arXiv version]{DBLP:conf/stoc/KempaK22} presented a deterministic implementation of a collection $\mathcal{X}$ of non-empty persistent strings, which allows to remove randomness from the statement of \cref{thm:dynamic} at the expense of replacing the $\log N$ factors with $\polylog N$. 

\paragraph*{Quantum setting}

We say an algorithm on an input of size $n$ succeeds \emph{with high probability} if the success probability can be made at least $1-1/n^c$ for any desired constant $c>1$.

In what follows, we assume the input strings can be accessed in a quantum query model~\cite{AMB04,DBLP:journals/tcs/BuhrmanW02}.
We are interested in the time complexity of our quantum algorithms~\cite{BBCplus}.

\begin{observation}[{\cite[Observation 2.3]{DBLP:conf/soda/JinN23}}]
For any two strings $S,T$ of length at most $n$,
$\lcp(S, T)$ or $\lcp^R(S, T)$ can be computed in $\cOtilde(\sqrt{n})$ time in the quantum model with high probability.
\end{observation}

\begin{fact}[Corollary of {\cite{DBLP:journals/jda/HariharanV03}}, cf~{\cite[Observation 39]{DBLP:conf/stacs/Charalampopoulos24}}]
For any two strings $S,T$ of length at most $n$, with $|T| \le 2|S|$,
$\mathsf{IPM}(S, T)$ can be computed in $\cOtilde(\sqrt{n})$ time in the quantum model with high probability.
\end{fact}

All other \pillar operations trivially take  
$\cO(1)$ quantum time. As a corollary, in the quantum setting, all \pillar operations can be implemented in $\cOtilde(\sqrt{m})$ quantum time with no preprocessing, as we always deal with strings of length $\cO(m)$.
Additionally, we can compute the linked list representation of the endpoints of groups of wildcards in $\cOtilde(\sqrt{m}G)$ time in the quantum model with high probability as follows: we search for a wildcard and, if we find it, we compute the group that contains it in $\cO(1)$ time in the \pillar model using \cref{fact:infty}; we then recurse on both sides of the group, and so on. In total, this procedure requires $\cO(G)$ time in the \pillar model and hence $\cOtilde(\sqrt{m}G)$ time in the quantum model.
As a corollary, we obtain the following result.

\begin{theorem}[Quantum Setting]\label{thm:quantum}
Consider a length-$m$ string $P$ with $D$ wildcards arranged in $G$ groups, a solid length-$n$ string $T$, and an integer threshold $k \geq 0$.
The $k$-mismatch occurrences of $P$ in $T$ can be computed in $\cOtilde((n/\sqrt{m})(G+k)(D+k))$ time in the quantum model with high probability.
\end{theorem}
\end{document}